\documentclass{article}

    \usepackage[preprint]{neurips_2025}

\usepackage{wrapfig}
\usepackage{subfigure}
\usepackage[ruled]{algorithm2e}
\usepackage[utf8]{inputenc} 
\usepackage[T1]{fontenc}    
\usepackage{hyperref}       
\usepackage{url}            
\usepackage{booktabs}       
\usepackage{amsfonts}       
\usepackage{nicefrac}       
\usepackage{microtype}      
\usepackage{xcolor}         
\usepackage{diagbox}
\usepackage{natbib}
\setcitestyle{numbers,square}

\usepackage{amsmath}
\usepackage{amssymb}
\usepackage{mathtools}
\usepackage{amsthm}
\usepackage{booktabs}
\usepackage{multirow}
\usepackage[normalem]{ulem}
\useunder{\uline}{\ul}{}
\usepackage{tcolorbox}
\usepackage{threeparttable}
\usepackage{captdef}
\usepackage{amsthm}
\usepackage[capitalize,noabbrev]{cleveref}
\usepackage{marvosym}

\theoremstyle{plain}
\newtheorem{theorem}{Theorem}[section]

\theoremstyle{definition}

\theoremstyle{remark}

\usepackage[textsize=tiny]{todonotes}

\title{Large Language Model as Universal Retriever in Industrial-Scale  Recommender System}

\author{Junguang Jiang\thanks{Equal contribution.}, Yanwen Huang\footnotemark[1], Bin Liu, Xiaoyu Kong, Xinhang Li, \\ \textbf{Ziru Xu, Han Zhu, Jian Xu, Bo Zheng}\textsuperscript{\Letter}  \\
    Taobao \& Tmall Group of Alibaba, China  \\ 
    {\tt\small {\{jiangjunguang.jjg,huangyanwen.hyw,zhuoli.lb,linlin.kxy,xinhangli.lxh,}} \\ 
    {\tt\small {ziru.xzr,zhuhan.zh,xiyu.xj,bozheng\}}@alibaba-inc.com}
}

\begin{document}

\maketitle

\begin{abstract}
  In real-world recommender systems, different retrieval objectives are typically addressed using task-specific datasets with carefully designed model architectures. We demonstrate that Large Language Models (LLMs) can function as universal retrievers, 
  capable of handling multiple objectives within a generative retrieval framework. 
  To model complex user-item relationships within generative retrieval, we propose multi-query representation. To address the challenge of extremely large candidate sets in industrial recommender systems, we introduce matrix decomposition to boost model learnability, discriminability, and transferability, and we incorporate probabilistic sampling to reduce computation costs. Finally, our Universal Retrieval Model (URM) can adaptively generate a set from tens of millions of candidates based on arbitrary given objective while keeping the latency within tens of milliseconds.
  Applied to industrial-scale data, 
  URM outperforms expert models elaborately designed for different retrieval objectives on offline experiments
   and significantly improves the core metric of online advertising platform by $3\%$.


\end{abstract}

\vspace{-5pt}
\section{Introduction}
\vspace{-5pt}
Recommender systems have become an integral part of people's daily lives, revolutionizing the way users interact with content and services.
To meet the diverse needs of users, recommender systems are evolving to become more sophisticated and multifaceted, leading to the emergence of a wide variety of objectives. For example, due to shifts in user behavior across contexts — such as different apps or interfaces within the same app — multi-scenario objectives are introduced \cite{MMOE}.
To accurately capture user intent, objectives like click prediction, purchase prediction, and favorite prediction are formulated \cite{cite:multi_objective_recommendation}. To avoid information cocoons and provide novel options, goals such as serendipity \cite{cite:discovery_recommendation, Serendipity} and long-tail item recommendations \cite{cite:longtail_recommendation} are established.

Most recommender systems use multi-channel retrieval to address the above issues \cite{huang2024comprehensivesurveyretrievalmethods, MultiChannel}. For each of the above objectives, a large amount of training data is collected, and specialized retrieval models are then designed, trained, evaluated, and deployed separately.
This approach has served well to make progress on narrow objectives, but when the objectives change, it requires collecting new data and training a new retrieval model carefully, which is time-consuming, lacks scalability, and sometimes faces challenges due to insufficient objective-specific data.
This pattern has more trouble when dealing with unclear objectives.  
In industry, a common unresolved issue in retrieval is that for a given online metric, such as advertising revenue, there is often no clear corresponding offline metric \cite{EBR_facebook}. 
The current strategy for tackling this issue is using human prior knowledge to develop multiple heuristic offline training objectives and then conducting  A/B tests to assess their effectiveness.


Another approach is learning multiple retrieval objectives simultaneously, which is known as multi-task learning \cite{Caruana1997MultitaskL}. However, this will easily lead to \textit{seesaw phenomenon} \cite{PLE}, where performance declines compared to learning objectives separately.
 Traditionally, objective conditioning is handled at the architectural level, where the parameter for each objective is manually designed according to its data size, and as the number of objectives continues to increase,
model design becomes increasingly challenging.
However, the emergence of LLMs offers a new approach in which objectives can be defined through text sequences \cite{cite:GPT,cite:GPT2,GPT3}.
Previous works \cite{kaplan2020scaling, hernandez2021scaling} indicate that increasing model size in such a structure can effectively reduce the seesaw phenomenon without handcrafted designs. 
Simultaneously, since text sequences can describe an infinite number of objectives,   when retrieval objectives are unclear, we can search within a larger objective space with just one trained model.

Despite the excellent versatility of LLMs, a key challenge to using them in retrieval is how to improve performance and efficiency while retaining versatility.
There are two main approaches. The first feeds both user features and item candidates to the LLMs for evaluation, which performs well on academic datasets, but is challenging to scale to billions of items in industrial systems \cite{DBLP:journals/corr/abs-2305-11700,bao2023bi,TallRec}. The second approach feeds user feature sequences to LLMs to generate target items, enhancing computational feasibility within the scope of generative retrieval.
Yet generative retrieval faces three inherent challenges:
(1) \textbf{Model Expressiveness}: 
While LLMs are highly expressive, generative retrieval struggles to effectively model user-item relationships because they generate target items using a single linear matrix \cite{zhu2018learning,zhu2019joint,gao2021deepretrievallearningretrievable}.
(2) \textbf{Learnability}: 
In industrial applications, the vast number of item candidates, often in the tens of millions, complicates learning the mapping from LLMs' hidden representations to the item candidate space due to the expansive parameter scale \cite{cite:TIGER}.
(3) \textbf{Efficiency}: Both LLM inference and generation within large candidates involve substantial computational costs.
By incorporating semantic IDs and generating autoregressively within a reduced space   \cite{DBLP:conf/nips/0001YCWZRCYRR23,DBLP:journals/corr/abs-2411-11739,DBLP:conf/cikm/0007BLZ0FNC24,DBLP:journals/corr/abs-2502-16474,DBLP:journals/corr/abs-2504-06780}, the above issues can be alleviated to some extent.
However, the adoption of LLMs in generative retrieval introduces challenges to deployment, particularly due to the significant increase in online computational costs.

In this paper, we introduce the \textbf{Universal Retrieval Model (URM)}, a novel pathway towards implementing generative retrieval. 
URM unifies various retrieval objectives into a cohesive input-output framework and utilizes LLMs as feature generators to enhance versatility, facilitating prompt tuning in industrial recommender systems.
By adopting multi-query representation, URM significantly boosts the expressive capacity of generative retrieval models. Matrix decomposition further enhances URM's learnability, discriminability, and transferability, while probabilistic sampling effectively reduces training and generation costs when faced with an extremely large candidate set.
Finally, URM is capable of dynamically adjusting retrieval outputs based on input instructions and efficiently returning collections of items with latency within tens of milliseconds. Our extensive experiments on both public and industrial-scale datasets demonstrate the robust performance and versatility of URM. Furthermore, online A/B tests on our advertising platform confirm the effectiveness of the URM in real-world settings, achieving improvements of over $3\%$ in the key metric.

\vspace{-5pt}
\section{Retrieval in Recommender Systems}
\vspace{-5pt}
In typical recommender systems, the main goal of retrieval is to identify the subset of the candidate set $\mathcal{C}$ that has the highest value according to a specific objective, i.e., $\mathrm{argTopk}_{v \in \mathcal{C}} f(u,v)$, 
where function $f(u,v)$ is an objective function for user $u$ and candidate item $v$.

\textbf{Embedding-Based Retrieval (EBR).}
To handle the large candidate set $\mathcal{C}$,  EBR \cite{EBR_facebook} is commonly used in practice, where the inner product of user and item representations is used to express $f(u,v)$. 
However, EBR falls short in capturing the intricate dynamics of user-item preferences \cite{zhu2018learning, gao2021deepretrievallearningretrievable, MoL}.
Consequently, researchers are actively developing techniques to address large-scale retrieval challenges using more sophisticated models, such as Model-Based Retrieval and Generative Retrieval.

\textbf{Model-Based Retrieval.} 
Advanced retrieval models  \cite{zhu2018learning,zhu2019joint, zhuo2020learning, gao2021deepretrievallearningretrievable, chen2022approximate} use 
multiple MLP layers to model $f(u,v)$. 
To avoid the substantial computational cost associated with scoring the entire candidate set, model-based retrieval uses techniques such as hierarchical tree structures \cite{zhu2018learning,zhu2019joint, zhuo2020learning} or learnable paths \cite{gao2021deepretrievallearningretrievable} to build indices that model user interests from coarse to fine detail.
Model-based retrieval demonstrates strong fitting ability to a given objective, but its scalability is constrained in practice.
For different objectives $f$, such as maximizing the probability of user interests (active behaviors like clicks, purchases, etc.)  or maximizing exposure probability across different scenarios, it is often necessary to design multi-task learning architectures specifically \cite{jiang2022adaptivedomainnetworkmultidomain,MMOE,PLE}
or even maintain separate retrieval models for each objective (multi-channel retrieval) \cite{huang2024comprehensivesurveyretrievalmethods}.

\textbf{Generative Retrieval.}
The success of generative AI \cite{cite:GPT,cite:GPT2, GPT3, gpt4} has led to an increasing amount of work focusing on using generative approaches for retrieval \cite{cite:TIGER,cite:IDGenRec,DBLP:journals/corr/abs-2502-18965}. 
The goal of generative retrieval is to identify the subset that maximizes conditional probability $\mathrm{argTopk}_{v \in \mathcal{C}} P(v|u)=\text{softmax}(W^T F(u))|_v$, 
where $F(u) \in R^D $ is a feature generator for user $u$ and $W\in R^{D \times |\mathcal{C}|}$ is a linear matrix that decodes hidden features into the candidate item space.
The benefit of generative retrieval is its better compatibility with Transformer architectures as well as LLMs, whose scalability has been thoroughly validated. 

Generative retrieval also has its own issues, such as  \textit{model expressiveness}. In model-based retrieval, the relationship between items and users is modeled using multi-layer MLPs, which can theoretically express arbitrary functions. In contrast, the generative retrieval method models the user-item relationship only through inner products, which significantly limits the model's capacity \cite{zhu2018learning,MoL,huang2024comprehensivesurveyretrievalmethods}.
Generative retrieval also has issues in
\textit{learnability} and \textit{efficiency}.
Unlike NLP, where the candidate is typically in the hundreds of thousands, industrial recommender systems deal with millions to billions of candidates.
Previous research has identified significant challenges in these high-dimensional spaces, pointing out that learning the $W$ matrix is particularly challenging and that operations like inner product and softmax incur substantial inference costs \cite{cite:TIGER,DBLP:conf/nips/0001YCWZRCYRR23,DBLP:journals/corr/abs-2410-09560,DBLP:journals/corr/abs-2410-19276}. 
Therefore, they introduce semantic IDs, which implicitly break down the candidate set $\mathcal{C}$. Specifically, $T$ semantic IDs $\{ \overline{v}_t \}_{t=1}^T$ are generated in an autoregressive manner, $\overline{v}_t = \arg\max\overline{W}^T F(u, \overline{v}_1, ..., \overline{v}_{t-1})$, where $\overline{W}$ is low-dimensional matrices, thus significantly reducing the optimization difficulty.
Then, the $T$ semantic IDs collectively determine the final item ID. 
However, some research found that semantic IDs often struggle with fine-grained similarity modeling (discriminability)
and cold-start problems (transferability) \cite{cite:COBRA,yang2024unifyinggenerativedenseretrieval,DBLP:conf/kdd/WangXHZ0LLLXZD24,DBLP:conf/recsys/ZhuJLQD024,DBLP:journals/corr/abs-2504-06636}.

In this work, we offer a new implementation of generative retrieval that enhances model expressiveness, learnability, and efficiency, and upgrade it to universal retrieval that surpasses the performance of model-based retrieval while offering superior task scalability.

\vspace{-5pt}
\section{Universal Retrieval Model}
\label{section:URM}
\vspace{-5pt}
Formally, the goal of universal retrieval model is to use a single model to identify the subset of the candidate set $\mathcal{C}$ that maximizes conditional probability for user $u$ under any given objective $o \in \mathcal{O}$, i.e., $\mathrm{argTopk}_{v \in \mathcal{C}} P(v|u,o)$, 
where $o$ might be encountered during training (multi-task learning), or may never appear during training (zero-shot learning). 
Universal retrieval demands higher model capacity, making generative retrieval more suitable due to its scalability and compatibility with LLMs. Specifically, $P(v|u,o)=\text{softmax}(W^T F(u, o))|_v. $ 
 Section \ref{sec: representations} will detail obtaining the universal representation $F(u, o)$ for a user $u$ and any objective $o$. Section \ref{sec:mapping} will tackle the learnability challenges of $W$ with tens of millions of candidates. Section \ref{sec:Multiplication} will focus on reducing the inference cost of $\text{softmax}(W^T F(u, o))$ using sampling techniques.

\subsection{Representations for  Users \& Any Objective}
\label{sec: representations}

Following the standard practice of Generative Pre-trained Transformers (GPTs) \cite{cite:GPT2}, we define different retrieval tasks in natural language and represent users and objectives in sequential form.

\textbf{User Description.} 
 To improve efficiency when handling long user behavior sequences, 
 URM treats items as a kind of special token following \cite{CoLLM,E4SRec,LlaRA,Adapt_LLM_Collaborative,cite:IDGenRec}.
Typical inputs are sentences composed of common text tokens and item IDs (such as [7502]) as follows:
 \begin{tcolorbox}[colback=blue!2!white,leftrule=2.5mm,size=title]
\label{obs1}
\footnotesize
The user attributes are as follows: age \{AGE\}, gender \{GENDER\}, located in \{PROVINCE\}, \{CITY\}. The user has favorited [7502]..., purchased [8274]..., clicked [8380].... 
\end{tcolorbox}
We have aggressively serialized all user features, including static attribute features, statistical features, and various behavioral sequence features. 
 In model-based retrieval methods, when a feature is missing, it is replaced with zeros, and sequence features often require truncation or padding to fit a fixed length. In contrast, our serialized approach omits the text description (e.g., \textit{user has purchased}) for any empty feature, and only truncates sequence features without padding, which reduces computational costs and enhances the model's robustness against varied inputs.

\textbf{Retrieval Objectives.} 
We define objectives using different text descriptions, such as: (1) Multi-scenarios: Please retrieve items for scenario A/B/C.
(2) Multi-behaviors: Please retrieve items that the user will click on / purchase / favorite.
 (3) Long-tail item: Please retrieve long-tail items.
 (4) Serendipity: Please retrieve items from new categories.
(5) Long-term interest: Please find items that match user's long-term interests.
 (6) Search: Please retrieve items that match the given query.
 We construct the training set with specific positive samples under these constraints to ensure that the model generates appropriate results for each objective.

 \begin{figure*}[!b]
    \centering
    \vspace{-25pt}
    \includegraphics[width=0.68\linewidth]{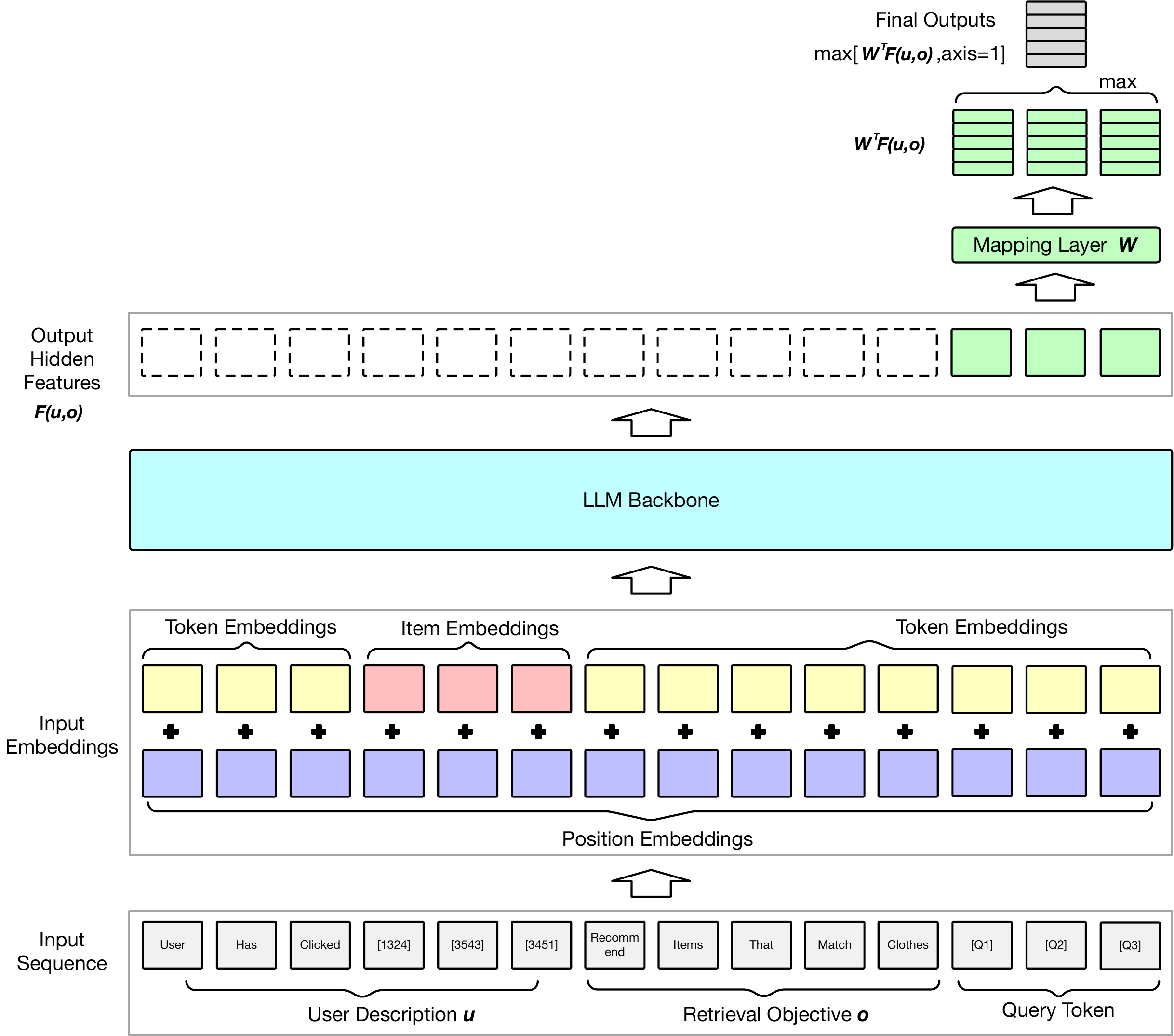}
    \vspace{-5pt}
    \caption{URM architecture. 
    The input sequence consists of user description $u$, retrieval objective $o$, and several fixed query tokens. 
    Item IDs in the user description are mapped to item embeddings by a distributed hashtable, and other tokens are mapped to token embeddings.
    The item embeddings or token embeddings are summed up with position embeddings, and fed into the LLM backbone (For LLMs using RoPE \cite{su2023roformerenhancedtransformerrotary},  position embedding is not explicitly added).
     The outputs corresponding to the query tokens are then mapped to the item candidate space through $W$.  To optimize for retrieval objectives in recommender systems, the parameters of LLM backbone are fully fine-tuned.
    }
    \label{fig:arch}
\end{figure*}

\textbf{Representations for Users \& Objectives.} In universal retrieval, the model must adapt its output to various user descriptions $u$ and retrieval objectives $o$, demanding high model capacity.   Thus, we use pre-trained LLMs as the feature generator $F$ due to their general capabilities across tasks \cite{gpt4}. 
Text tokens are mapped to embeddings via the LLM's vocabulary embedding table. With industry item IDs reaching billions, we employ a separate Distributed HashTable \cite{ParameterServer} to convert item IDs into unique embeddings, and then use a one-layer MLP to map these embeddings to the same dimension as the token embedding.  As shown in Figure \ref{fig:arch}, 
the embedding at each position in the sequence is then the sum of the position embedding and either the token or item embedding.
These embeddings are then processed through the LLMs to generate the features $F(u, o)$.

\textbf{Multi-Query Representation.}
Generative retrieval faces challenges in expressing user-item relationships effectively. Function approximation theory suggests that a linear combination of appropriately chosen basis functions or feature representations can approximate certain classes of complex functions \cite{HANDSCOMB196661}. Inspired by this conclusion, we propose adding 
multiple special query tokens to the end of the input sequence, allowing the LLMs to generate $M$ user representations $\mathbf{F}(u,o)\in R^{D\times M}$ at these positions during a single LLM forward process. These representations are then used to compute inner products with $W\in R^{D\times |\mathcal{C}|}$, and their max value is used as the final score, i.e. $P(v|u,o)=\text{softmax}(\max(W^T F(u, o), \text{axis}=1))|_v$
 (in practice, we found that the max function works better than linear combinations). This approach fully preserves the model's ability to capture complex interactions between users and items, thus enhancing the model's potential performance.
Note that while the query tokens are fixed, their embeddings are learnable. Besides, the forward attention inherent in LLMs allows the hidden features of different query tokens to access the features from preceding query tokens, which enables representations of different query tokens to capture different user interests (see Appendix \ref{appendix:multi_query_representation} for visualizations).

\vspace{-5pt}
\subsection{Mapping to Large-scale Candidate Sets}
\vspace{-5pt}

\label{sec:mapping}
Using LLMs as feature generators offers strong capability but also increases computational costs. For example, with a sequence length of $1024$, a single forward pass of a $1.5$B model can achieve a latency of $50$ ms, which is the limit for our online recommendation retrieval systems. 
Generative retrieval methods using semantic IDs require multiple calls to the feature generator, making the computational cost prohibitive when using  LLMs as feature generators.
This necessitates exploring a new implementation for generative retrieval, one that only requires a single LLM forward pass. 

Recall that the primary aim of introducing semantic IDs is to address the issue that matrix $W$ is too large and difficult to learn.
We propose an alternative solution to improve its learnability using matrix decomposition 
$W=UV^T$, where the shapes of $U$ and $V$ are $D \times H$ and $|\mathcal{C}| \times H$ respectively, and $H$ is a lower rank dimension. This approach reduces the parameter size, thereby facilitating easier learning and significantly decreasing computational demands.

Further, we can improve the mapping matrix $W$  to incorporate both discriminability and transferability. 
In retrieval problems, discriminability refers to the model's ability to accurately rank two relatively similar items under a given objective.  
When there is sufficient training data for certain items, assigning distinct parameters for these items is the most effective way to increase discriminability. We denote these distinct parameters as $V_{\text{dis}}$, which correspond to the fully learnable item-side representations during training.
Transferability refers to the retrieval model's ability to identify items that have not been seen in the training set based on their generalized features. In industrial recommender systems, online item candidate sets are constantly changing every minute, making transferability crucial for the retrieval model.
To improve transferability, we first serialize the generalized features of each item, including static attributes such as title and category, as well as statistical features like sales and click-through rates, into text. Examples are as follows: 
 \begin{tcolorbox}[colback=blue!2!white,leftrule=2.5mm,size=title]
\label{obs2}
\footnotesize
The item title is \{TITLE\}. The category is \{CATEGORY\}. The price is \{PRICE\}. The shop name is \{SHOP\}. Over the past 7 days, its sales volume is \{SALE\}  and click-through rate is \{CTR\}.
\end{tcolorbox}
This serialized text is then fed into a General Text Embedding \cite{zhang2024mgte,li2023towards} LLM to obtain a fixed high-dimensional LLM representation. Then, a learnable linear layer is applied to produce a lower-dimensional matrix $V_{\text{trans}}$.
Then final mapping matrix is  $W = U(V_{\text{dis}}+V_{\text{trans}})^T$, 
where the total learnable parameters are $(|\mathcal{C}|+2D) \times H$. For items that never appear in the training set, the mapping matrix $W = UV_{\text{trans}}^T$ allows for the retrieval of cold-start items through transferability.

\vspace{-5pt}
\subsection{Approximation of Large Matrix Multiplication}
\vspace{-5pt}
\label{sec:Multiplication}
The learnability issue of $W$ can be solved by matrix decomposition, yet the computation costs remain high due to the high dimension of $|\mathcal{C}|$.
During training, we can approximate the negative log-likelihood using Noise Contrastive Estimation (NCE) Loss \cite{gutmann2010noise}, which avoids large matrix operations and significantly accelerates the training process,
\begin{equation}
 \min \mathcal{L}_{\text{NCE}} (u, o) =   - \sum_{v \in \mathcal{P}(u, o)} \log \dfrac{\exp [\max(W_v^T\mathbf{F}(u,o))]}{ \sum_{z\in \{v\} \cup \mathcal{N}} \exp [\max(W_z^T\mathbf{F}(u,o))]} ,
\end{equation}
where $\mathcal{P}(u, o)$ denotes the set of positive samples for user $u$ under objective $o$ and  $\mathcal{N}$ represents negative examples sampled from the item candidate set according to their occurrence frequency.

During inference, we can also sample a subset  $S$ for retrieval,
\begin{equation}
    \text{TopK}_{v \in \mathcal{S}} P(v|u,o)= \text{softmax}(\max(W_S^T \mathbf{F}(u,o), \text{axis}=1) )|_v
\end{equation}
where $W_S$ refers to the sub-matrix of  $W$ corresponding to the subset $S$.
Although this method can reduce computational cost,  many potential items of user interest might never be retrieved.
 To tackle this issue, 
 we employ probability sampling repeatedly, as detailed in Algorithm  \ref{algorithm}. 
 Specifically, 
 we construct an ANN index \cite{HNSW} for $W$ that allows each item $s$ to locate its neighbors  $\text{NBR}(s)$. Initially, a fixed subset is randomly selected from the candidate set, and item probabilities are computed using the relevant sub-matrix.  Then $K$ items are sampled based on their probabilities, forming a new subset influenced by their neighbors. Since probabilities are calculated on a subset shaped by the previous step's outcomes, this method resembles a 1-gram autoregressive model. The underlying assumption of Algorithm \ref{algorithm} is that items close in $W$ matrix will also have similar $P(v|u,o)$. Appendix \ref{proof} provides a strict theoretical guarantee and its proof.

\begin{wrapfigure}{R}{0.45\textwidth}
\vspace{-15pt}
\begin{minipage}{0.45\textwidth}
    \begin{algorithm}[H]
\scriptsize
    \caption{Probabilistic Sampling.}
    \label{algorithm}
    
    $N(0) = \text{Subset}(\mathcal{C})$
    
    $\hat{\mathbf{F}}(u, o) = U^T \mathbf{F}(u, o)$
    
    \For{$t\leftarrow 1$ \KwTo $T$}{
        $P(v|u,o) = \text{softmax}[{\max(V_{N(t-1)} \hat{\mathbf{F}}(u,o), \text{axis}=1)}/{\tau}]|_v$ \\
        
        $S(t) \sim \text{Sample}_{v \in N(t-1)}^K P(v|u,o)$ \\
        
        $N(t) =  \cup_{s\in S(t)} (\text{NBR}(s) \cup s) $
    }

    \textbf{return} $S(T)$
\end{algorithm}
\end{minipage}
\vspace{-15pt}
\end{wrapfigure}
 
    
    
    
        
        


    
    
        
        


Here, we follow the strategy used in NLP by employing probabilistic sampling instead of TopK selection \cite{holtzman2020curiouscaseneuraltext,wang2023selfconsistencyimproveschainthought}. 
Theoretically, probabilistic sampling provides a fairer treatment for items ranked at positions $K$ and $K+1$.
As the temperature  $\tau$ approaches $0$, the retrieval result of probabilistic sampling will gradually approximate that of TopK selection. Besides, we use the same $\mathbf{F}(u,o)$ in each probability calculation, thus it does not increase the inference cost of LLMs.
After optimization, the complexity of matrix computation is $\mathcal{O}(MH(D+TK\times\max(|\text{NBR}(\cdot)|))$, while using the full $W$ for matrix computation incurs a cost of $\mathcal{O}(MD|\mathcal{C}|)$.
In practice, we have $\tau=0.07, M=128, T=4, H=128, D=4096, K=1000, \max(|\text{NBR}(\cdot)|)=32, |\mathcal{C}|\approx 10^7$, thus, the FLOPS is reduced from $5000G$  to $2G$ 
, relatively small compared to LLMs.

\begin{wrapfigure}{r}{0.45\textwidth}
\vspace{-25pt}
    \centering
    \includegraphics[width=1\linewidth]{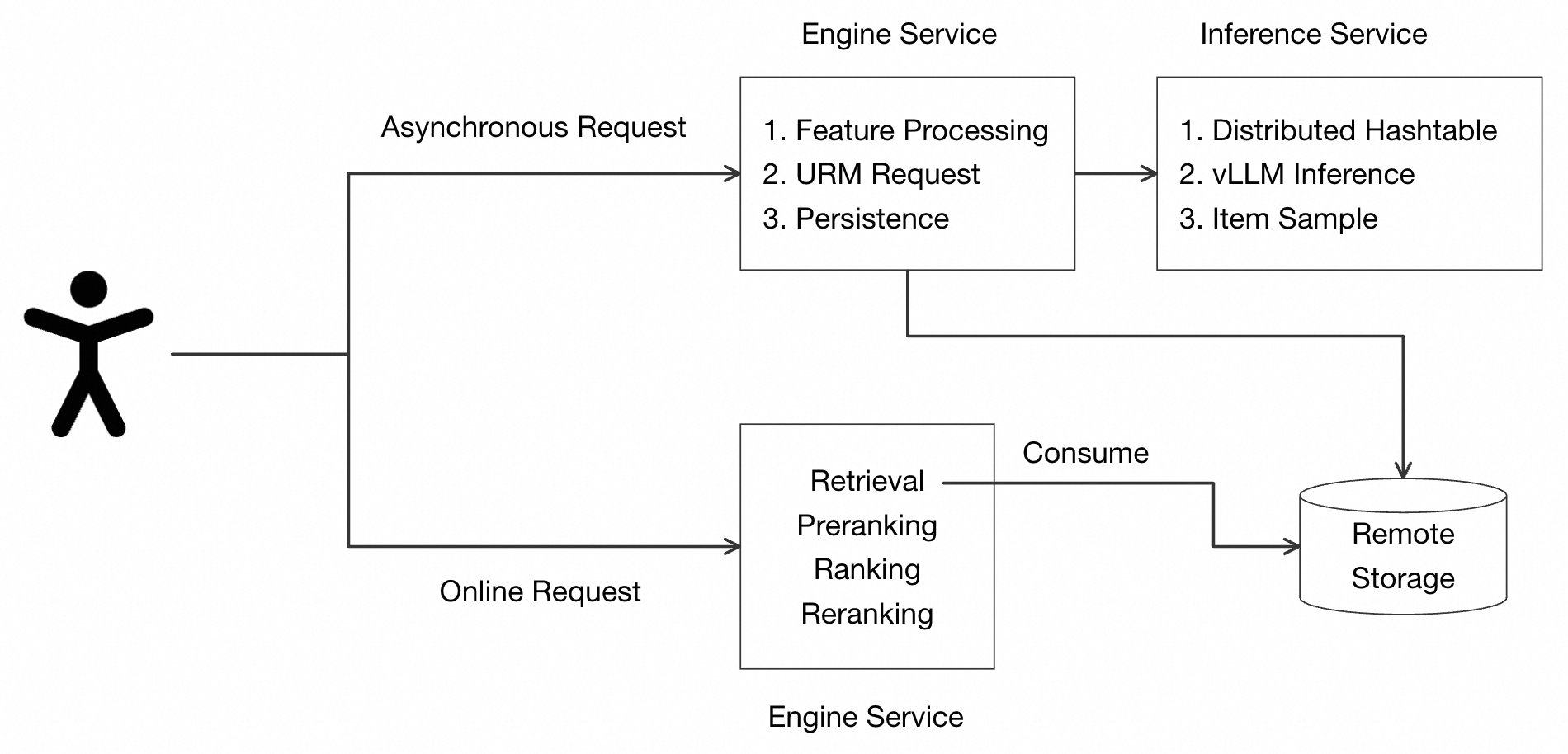}
    \vspace{-15pt}
    \caption{Online Serving System. }
    \vspace{-12pt}
    \label{fig:system_impl}
\end{wrapfigure}

\vspace{-8pt}
\section{System Implementation}
\vspace{-8pt}

Figure \ref{fig:system_impl} illustrates the online serving system of  URM. 
Considering the latency and computational cost of LLMs, we develop an asynchronous workflow to meet online requests. URM inference is triggered asynchronously in response to user actions, such as clicks and purchases.
Then the engine service reads user features and concatenates them with the retrieval objective texts, and sends a request to the inference service. The inference service reads item embeddings from a distributed hashtable, concatenates them with token embeddings, and then calls vLLM \cite{vllm} for inference (a single prefill operation). Finally, the sampling technique is used to generate the final retrieval results, which are then stored persistently for use during the online retrieval phase.
By aggregating user behavior within a configurable window (we use ten-minute online) to trigger a URM inference, the computational resources can be reduced to those needed for model-based retrieval.

\vspace{-8pt}
\section{Experiments}
\vspace{-8pt}
We validate URM's effectiveness using $4$ public datasets (Section \ref{sec:open_dataset_experiment}), an industrial-scale offline dataset (Section \ref{sec:industrial_experiment}), and online A/B tests (Section \ref{sec:online_results}). Ablation studies are presented in Section \ref{sec:ablation_study}, while universal retrieval performance is assessed in Section \ref{sec:universal_retrieval_performance}. 
We use Qwen-7B \cite{qwen} as the feature generator by default, and the results for other LLMs can be found in \ref{appendix:more_backbone}. The inference efficiency is given in Appendix \ref{appendix:inference_efficiency}.
Our code is based on LLaMA-Factory \cite{zheng2024llamafactory} and we will release our code related to the public datasets upon acceptance.

\vspace{-5pt}
\subsection{Public Dataset Experiments}
\vspace{-5pt}
\label{sec:open_dataset_experiment}
\textbf{Dataset.} We first evaluate the performance of URM on four widely-recognized public datasets: Sports \& Outdoors, Beauty and Toys \& Games \cite{amazon}, and Yelp \cite{yelp}. 
We follow the preprocessing methods used in recent works \cite{P5, E4SRec, cite:S3Rec} to construct the training and test datasets.

\textbf{Models.}
We use multiple retrieval methods as baseline:
(1) HGN \cite{cite:HGN} leverages a hierarchical structure to model user-item interactions.
(2)  GRU4Rec \cite{cite:GRU4Rec} 
    employs GRUs to model sequential user behavior.
(3)  Caser \cite{cite:Caser}
    leverages CNNs to capture sequential patterns in user behavior. 
(4)    BERT4Rec \cite{cite:Bert4Rec} adapts the BERT architecture for sequential recommendation.
(5) FDSA \cite{FDSA}
applies a self-attention module to model the relationships between features.
 (6)   SASRec \cite{cite:SasRec} employs self-attention mechanisms to model user behavior sequences.
 (7)   S3-Rec \cite{cite:S3Rec}
    enhances sequential recommendation by incorporating self-supervised learning.
 (8)   E4SRec \cite{E4SRec}  integrates item IDs within  LLMs \cite{Llama}.
 (9)   P5 \cite{cite:P5}  employs the T5 \cite{T5} model and takes recommendation tasks as purely natural language tasks.
 (10)  TIGER \cite{cite:TIGER} introduces semantic IDs into  generative retrieval.
 (11)  IDGenRec \cite{cite:IDGenRec} utilizes LLMs to create semantically rich and unique textual IDs for items.
 (12)  COBRA \cite{cite:COBRA} incorporates a cascaded sparse-dense representation framework to integrate sparse semantic IDs with dense vectors.
 The user behavior sequence length is limited to $100$.

 The performance is evaluated using  Hit Rate (HR@$K$) and Normalized Discounted Cumulative Gain (NDCG@$K$), computed at different ranking positions.
As presented in Table \ref{table:public_k5}, 
URM outperforms the strongest baseline by an average of $\textbf{46\%}$ and $\textbf{29\%}$ in terms of HR@5 and NDCG@5, relatively.
\begin{table}[!t]
\vspace{-5pt}
\addtolength{\tabcolsep}{2pt}
\centering
\scriptsize
\caption{Performance on $4$ public datasets (Only results for $K= 5$ are shown here, while those for $K = 10$ are available in Appendix \ref{appendix:more_public_k}). RI: Relative Improvement.}
\vspace{-5pt}
\begin{tabular}{@{}l|cc|cc|cc|cc@{}}
\toprule
\multicolumn{1}{c|}{\multirow{2}{*}{\textbf{Methods}}} & \multicolumn{2}{c|}{\textbf{Sports}} & \multicolumn{2}{c|}{\textbf{Beauty}} & \multicolumn{2}{c|}{\textbf{Toys}} & \multicolumn{2}{c}{\textbf{Yelp}} \\ \cmidrule{2-9} 
\multicolumn{1}{c|}{}                                  & HR@5              & NDCG@5           & HR@5              & NDCG@5           & HR@5             & NDCG@5          & HR@5            & NDCG@5          \\ \midrule
HGN                                                    & 0.0189            & 0.0120           & 0.0325            & 0.0206           & 0.0321           & 0.0221          & 0.0186          & 0.0115          \\
GRU4Rec                                                & 0.0129            & 0.0086           & 0.0164            & 0.0099           & 0.0097           & 0.0059          & 0.0152          & 0.0099          \\
Caser                                                  & 0.0116            & 0.0072           & 0.0205            & 0.0131           & 0.0166           & 0.0107          & 0.0151          & 0.0096          \\
BERT4Rec                                               & 0.0115            & 0.0075           & 0.0203            & 0.0124           & 0.0116           & 0.0071          & 0.0051          & 0.0033          \\
FDSA                                                   & 0.0182            & 0.0122           & 0.0267            & 0.0163           & 0.0228           & 0.0140          & 0.0158          & 0.0098          \\
SASRec                                                 & 0.0233            & 0.0154           & 0.0387            & 0.0249           & 0.0445           & 0.0236          & 0.0162          & 0.0100          \\
S3-Rec                                                 & 0.0251            & 0.0161           & 0.0387            & 0.0244           & 0.0443           & 0.0294          & 0.0201          & 0.0123          \\
E4SRec                                                 & 0.0281            & 0.0196           & 0.0525            & 0.0360           & 0.0566           & 0.0405          & 0.0266          & 0.0189          \\
P5                                                     & 0.0387            & 0.0312           & 0.0508            & 0.0379           & 0.0648           & {\ul 0.0567}    & {\ul 0.0574}    & {\ul 0.0403}    \\
TIGER                                                  & 0.0264            & 0.0181           & 0.0454            & 0.0321           & 0.0521           & 0.0371          & -               & -               \\
IDGenRec                                               & {\ul 0.0429}      & {\ul 0.0326}     & {\ul 0.0618}      & {\ul 0.0486}     & {\ul 0.0655}     & 0.0481          & 0.0468          & 0.0368          \\
COBRA                                                  & 0.0305            & 0.0215           & 0.0537            & 0.0395           & 0.0619           & 0.0462          & -               & -               \\
URM                                                    & \textbf{0.0733}   & \textbf{0.0488}  & \textbf{0.0929}   & \textbf{0.0671}  & \textbf{0.0888}  & \textbf{0.0619} & \textbf{0.0724} & \textbf{0.0476} \\ \midrule
\textbf{RI}                                                     & +70.9\%           & +49.7\%          & +50.3\%           & +38.1\%          & +35.6\%          & +9.2\%          & +26.1\%         & +18.1\%         \\ \bottomrule
\end{tabular}
\label{table:public_k5}
\vspace{-5pt}
\end{table}

\begin{table}[!t]
\addtolength{\tabcolsep}{0.5pt}
\centering
\scriptsize
\caption{Performance on the industrial dataset (metric: R@1000).}
\vspace{-5pt}
\begin{tabular}{@{}llllllllllll@{}}
\toprule
\textbf{Model} &
  \textbf{Learning Method} &
  \textbf{CPR} &
  \textbf{RSA} &
  \textbf{RSB} &
  \textbf{RSC} &
  \textbf{SR} &
  \textbf{LR} &
  \textbf{LIR} &
  \textbf{PPR} &
  \textbf{RQ} &
  \textbf{AVG} \\ \midrule
\multirow{2}{*}{\begin{tabular}[c]{@{}l@{}}Two-tower \\ Model\end{tabular}} &
  STL &
  0.129 &
  0.271 &
  0.166 &
  0.129 &
  0.069 &
  0.066 &
  0.117 &
  0.146 &
  0.355 &
  0.161 \\
 & MTL          & 0.120       & 0.205 & 0.166       & 0.135       & 0.064 & 0.115          & 0.103          & 0.173       & 0.257       & 0.149       \\ \midrule
\multirow{2}{*}{\begin{tabular}[c]{@{}l@{}}Transformer-\\ based Model\end{tabular}} &
  STL &
  0.198 &
  0.409 &
  0.293 &
  0.208 &
  {\ul 0.104} &
  0.115 &
  0.213 &
  0.143 &
  0.593 &
  0.253 \\
 & MTL          & 0.192       & 0.390 & 0.319       & 0.221       & 0.076 & 0.218          & 0.207          & 0.401       & 0.744       & 0.308       \\ \midrule
\multirow{5}{*}{\begin{tabular}[c]{@{}l@{}}Attention-\\ DNN\end{tabular}} &
  STL &
  0.253 &
  {\ul 0.477} &
  0.338 &
  0.260 &
  \textbf{0.106} &
  0.213 &
  {\ul 0.251} &
  0.353 &
  0.651 &
  0.323 \\
 & MTL          & 0.238       & 0.456 & 0.375       & {\ul 0.277} & 0.062 & {\ul 0.336}    & \textbf{0.265} & 0.478       & 0.671       & 0.351       \\
 & MTL-SharedBottom & 0.243       & 0.442 & 0.376       & 0.270       & 0.072 & \textbf{0.337} & 0.224          & {\ul 0.505} & 0.745       & 0.357       \\
 & MTL-MMoE         & 0.233       & 0.439 & 0.375       & 0.257       & 0.070 & 0.325          & 0.218          & 0.491       & 0.736       & 0.349       \\
 & MTL-PLE          & {\ul 0.256} & 0.451 & {\ul 0.397} & 0.274       & 0.062 & 0.327          & 0.224          & 0.512       & {\ul 0.761} & {\ul 0.363} \\ \midrule
 URM & 
  MTL &
  \textbf{0.263} &
  \textbf{0.530} &
  \textbf{0.439} &
  \textbf{0.362} &
  0.093 &
  0.285 &
  0.240 &
  \textbf{0.581} &
  \textbf{0.835} &
  \textbf{0.403} \\ \bottomrule
\end{tabular}
\label{table:taobao}
\vspace{-15pt}
\end{table}

\vspace{-5pt}
\subsection{Industrial-scale Experiments}
\vspace{-5pt}
\label{sec:industrial_experiment}
\textbf{Dataset.}   Next, we verify the effectiveness of URM in an industrial-scale dataset, which is obtained from real traffic logs of the online system. 
The typical objectives include \textbf{C}lick \textbf{P}rediction \textbf{R}etrieval (\textbf{CPR}), \textbf{R}etrieval for \textbf{S}cene \textbf{A} (\textbf{RSA}), \textbf{S}cene \textbf{B} (\textbf{RSB}), \textbf{S}cene \textbf{C} (\textbf{RSC}), \textbf{S}erendipity  \textbf{R}etrieval (\textbf{SR}), \textbf{L}ong-term \textbf{R}etrieval (\textbf{LR}), \textbf{L}ong-tail \textbf{I}tem \textbf{R}etrieval (\textbf{LIR}),  \textbf{P}urchase \textbf{P}rediction \textbf{R}etrieval (\textbf{PPR}) , and  \textbf{R}etrieval with \textbf{Q}uery (\textbf{RQ}). 
This dataset contains hundreds of millions of samples and more than one billion distinct items in user behavior sequences. 
The candidate set contains tens of millions of items.
A more detailed definition of the dataset is in Appendix \ref{appendix:dataset_details}.
We use samples from day $1$ to day $\mathcal{T}$ for training and the samples of day $\mathcal{T}+1$ for testing.
To improve training efficiency, \textit{we aggregate each user's daily behaviors as the target set}, which makes these retrieval tasks more challenging.
All methods are compared in this setting to facilitate a more effective result.

\textbf{Models.} For baseline, we use the most commonly used methods in practice.  (1) Two-tower Model, which uses two multi-layer MLPs to convert user-side and item-side features into embeddings, and uses the inner product to obtain the final score \cite{EBR_facebook}.  
(2) Transformer-based Model,  which uses a multi-layer transformer to convert user behavior sequence features into user embeddings \cite{cite:HSTU,liu2024kuaiformertransformerbasedretrievalkuaishou}.
(3) Attention-DNN, which calculates cross-attention between the user behavior sequence and the target, and then uses multiple layers to calculate the final score. 
For each method, we implement a single-task learning version (STL) with only one objective and a multi-task learning version (MTL) with multiple objectives. For the Attention-DNN model, which has the best performance among them, we further provide the Shared Bottom version \cite{ESMM}, the MMoE version \cite{MMOE}, and the PLE version \cite{PLE}. More details can be found in Appendix \ref{sec:Implementation}.
We tuned the hyperparameters of each method, including the embedding size and number of layers, using the validation set. 
The user behavior sequence length is $300$, and the total length of text tokens, query tokens, and item IDs in URM is truncated to $1024$.

We use recall to evaluate the effectiveness of our proposed method. Denote the output item set as $\mathcal{P}$  and the ground truth as $\mathcal{G}$,  then recall $\text{R@}K= {|\mathcal{P} \cap \mathcal{G} |} / {|\mathcal{G}|} , \text{where} 
 |\mathcal{P}|=K$. 
We show results in Table \ref{table:taobao}.  
Due to the complex task relationships, it is difficult for a single approach to perform well on all objectives, which leads to multi-task training being sometimes effective and sometimes causing seesaw phenomenon.
In contrast, our proposed URM achieves the best performance in $6$ out of $9$ objectives and achieves a relative improvement of $\textbf{11.0\%}$ in average across all objectives.
Besides, for RQ (Retrieval with Query), there is a specifically designed three-tower model, which adds an independent tower for queries and reaches $0.822$ on R@1000. In contrast, our model achieves comparable results without any specific design. More results can be found in Appendix \ref{appendix:more_prompt_engineering}.

\vspace{-5pt}
\subsection{Ablation Study}
\vspace{-5pt}

\label{sec:ablation_study}

\textbf{Multi-Query Representation.}
 As shown in Figure \ref{fig:MTP}, as the number of query tokens increases, URM's ability to represent the target item set becomes stronger, resulting in better performance.

\textbf{Matrix Decomposition.}
As demonstrated in Table \ref{table:item_embedding},  the discriminability matrix $V_{\text{dis}}$ significantly enhances the model’s retrieval capabilities across all items, while the transferability matrix $V_{\text{trans}}$ boosts the model's ability to retrieve items unseen during training. This underscores the effectiveness of matrix decomposition and highlights the importance of different item representations.
Appendix \ref{appedix:visualize_item} further provides a visualization comparison of different item representations.

\begin{figure}[htbp] 
\vspace{-10pt}
\scriptsize
\begin{minipage}[!b]{0.48\textwidth} 
        \centering
    \caption{The effect of query token number $M$.}
    \includegraphics[width=1\linewidth]{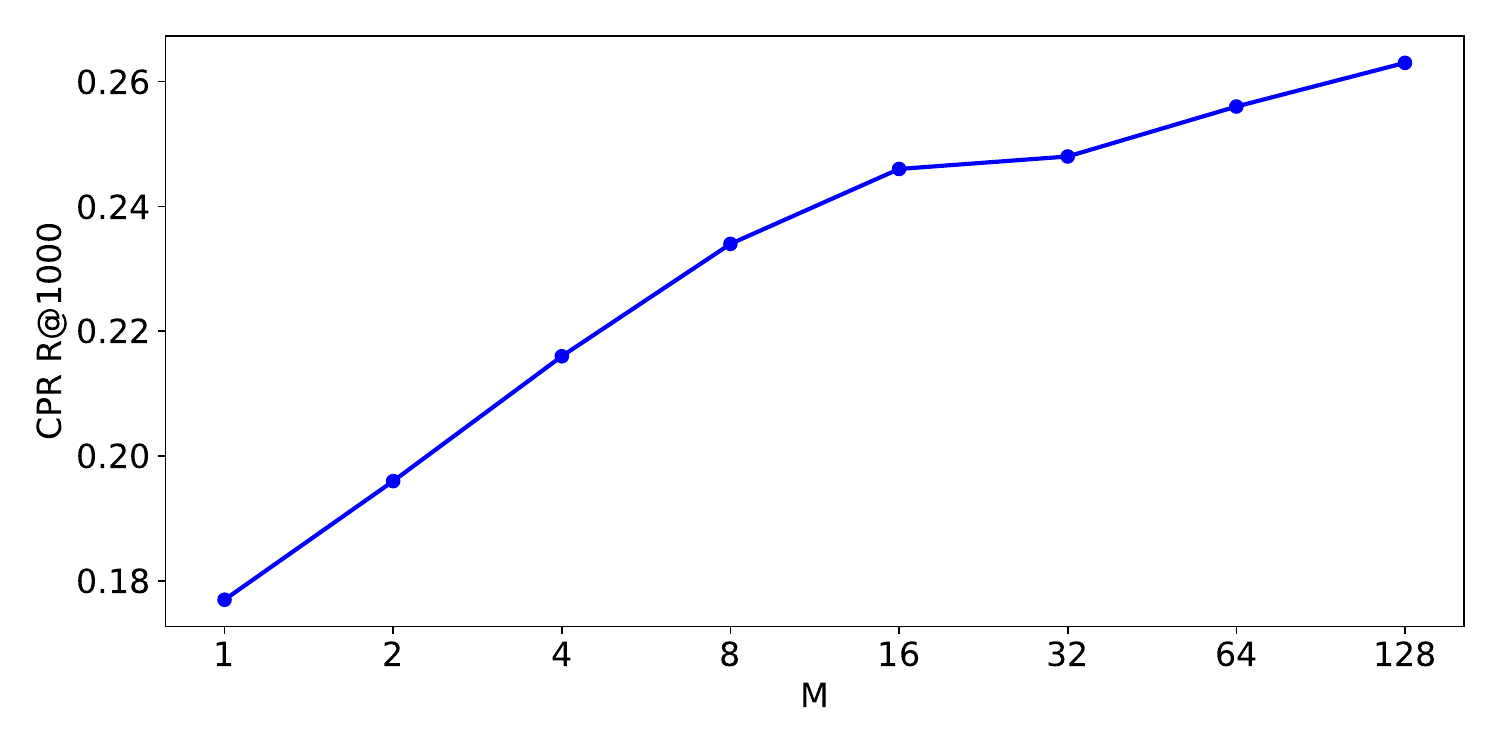}
    \label{fig:MTP}
\end{minipage}
\hspace{5pt}
\begin{minipage}[!b]{0.24\textwidth} 
    \vspace{-28pt}
\addtolength{\tabcolsep}{0pt}
\centering
\tabcaption{The effect of item representation $V$ in matrix decomposition on CPR R@1000.}
\begin{tabular}{@{}lll@{}}
\toprule
\textbf{$V$} & \begin{tabular}[c]{@{}l@{}}\textbf{All}\end{tabular} & \begin{tabular}[c]{@{}l@{}}\textbf{Unseen} \end{tabular} \\ \midrule
$V_{\text{dis}}$   & 0.256                                                    & 0.116                                                     \\
$V_{\text{trans}}$ & 0.152                                                     & 0.101                                                    \\
$V_{\text{dis}}+V_{\text{trans}}$                & \textbf{0.263}                                                    & \textbf{0.130}                                                       \\ \bottomrule
\end{tabular}
\label{table:item_embedding}
\end{minipage}
\hspace{5pt}
\begin{minipage}[!b]{0.24\textwidth} 
    \vspace{-30pt}
\addtolength{\tabcolsep}{8pt}
\centering
\tabcaption{The effect of sampling steps $T$.}
\begin{tabular}{@{}cc@{}}
\toprule
 T & Recall Precision \\
 \midrule
 \textbf{1}   &  0.2\% \\ 
 \textbf{2}   &  2.1\% \\ 
 \textbf{3}    &  41.7\% \\ 
 \textbf{4}    & 91.0\% \\
 \textbf{5}  & 91.1\%\\
\bottomrule
\end{tabular}
\label{table:retrieval}
\end{minipage} 
\vspace{-20pt}
\end{figure}

\textbf{Probabilistic Sampling.}
We use probability sampling to approximate the probability calculation for the entire candidate set. As the number of sample steps $T$ increases, the retrieval precision ($\frac{\text{R}@1000 \text{ w/ sampling}}{\text{R}@1000 \text{ w/o sampling}}$) continually increases and tends to converge when $T = 4$.

\vspace{-5pt}
\subsection{Universal Retrieval Performance}
\vspace{-5pt}
\label{sec:universal_retrieval_performance}
The performance of universal retrieval can be assessed from two perspectives: firstly, its effectiveness in handling objectives encountered during training (in the multi-task learning setting), and secondly, its transferability to objectives not seen during training (in the zero-shot learning setting). Only several examples are provided here, and more examples can be found in Appendix \ref{appendix:more_prompt_engineering} and \ref{appendix:zero_shot_prompt}.

\textbf{Multi-Task Learning.}
Although the candidates and generation are both different from pretrained LLMs, we observe that  URM remains sensitive to the text input. 
As shown in Table \ref{table:multi_scenario_prompt}, with objectives for specific scenarios, URM aligns more closely with the data distribution and achieves higher performance,
with a relative improvement of over $20\%$. 
As illustrated in Table \ref{table:serendipity_prompt}, 
after using the SR (Serendipity Retrieval) objective, the percent of the new category increased from $18.8\%$ to $46.2\%$, leading to a relative increase of $82.3\%$ in SR R@1000.

\begin{table}[htbp]
    \centering
\vspace{-5pt}
    \caption{Universal retrieval performance.
    }
    \vspace{-10pt}
    \subtable[multi-scenario retrieval]{
        \addtolength{\tabcolsep}{-4pt}
        \centering
        \scriptsize
        \begin{tabular}{llll}
        \toprule
        \textbf{Objective} & \textbf{\begin{tabular}[c]{@{}l@{}}RSA\\ R@1000\end{tabular}}   & \textbf{\begin{tabular}[c]{@{}l@{}}RSB\\ R@1000\end{tabular}}    & \textbf{\begin{tabular}[c]{@{}l@{}}RSC\\ R@1000\end{tabular}}    \\ \midrule
        CPR                           & 0.440          & 0.335          & 0.278          \\
        RSA                           & \textbf{0.530} & 0.409          & 0.240          \\
        RSB                           & 0.522          & \textbf{0.439} & 0.257          \\
        RSC                           & 0.444          & 0.327          & \textbf{0.362} \\ 
        \midrule
        \textbf{RI}                          & \textbf{+20.5\%}          & \textbf{+31.0\%}          & \textbf{+30.2\%} \\ 
        \bottomrule
        \end{tabular}
        \label{table:multi_scenario_prompt}
    }
    \subtable[serendipity retrieval]{
        \addtolength{\tabcolsep}{-4pt}
        \centering
        \scriptsize
        \begin{tabular}{@{}llll@{}}
        \toprule
        \textbf{Objective} &
          \textbf{\begin{tabular}[c]{@{}l@{}}CPR\\ R@1000\end{tabular}} &
          \textbf{\begin{tabular}[c]{@{}l@{}}SR\\ R@1000\end{tabular}} &
          \textbf{\begin{tabular}[c]{@{}l@{}}Percent of \\ New Category\end{tabular}} \\ \midrule
        CPR                    & \textbf{0.263} & 0.051            & 18.8\%            \\
        SR                    & 0.213          & \textbf{0.093}   & \textbf{46.2\%}   \\
        \midrule
        \textbf{RI} &          -       & \textbf{+82.3\%} & \textbf{+145.7\%} \\ \bottomrule
        \end{tabular}
        \label{table:serendipity_prompt}
    }
    \subtable[hybrid objectives]{
        \scriptsize
        \addtolength{\tabcolsep}{-4pt}
        \centering
        \scriptsize
        \begin{tabular}{@{}lll@{}}
        \toprule
        \textbf{Objective} & \textbf{\begin{tabular}[c]{@{}l@{}}RQ\\ R@1000\end{tabular}} & \textbf{\begin{tabular}[c]{@{}l@{}}Percent of \\ Long-tail Items\end{tabular}} \\ \midrule
        RQ        & 0.835          & 79.6\%          \\
        LIR       & 0.630          & 81.6\%          \\
        RQ $\times$ LIR & \textbf{0.836} & \textbf{82.4\%} \\ \bottomrule
        \end{tabular}
        \label{table:mix_SR_LIR}
    }
    \vspace{-10pt}
\end{table}

\begin{wrapfigure}{r}{0.42\textwidth}
\vspace{-37pt}
  \centering
    \includegraphics[width=1\linewidth]{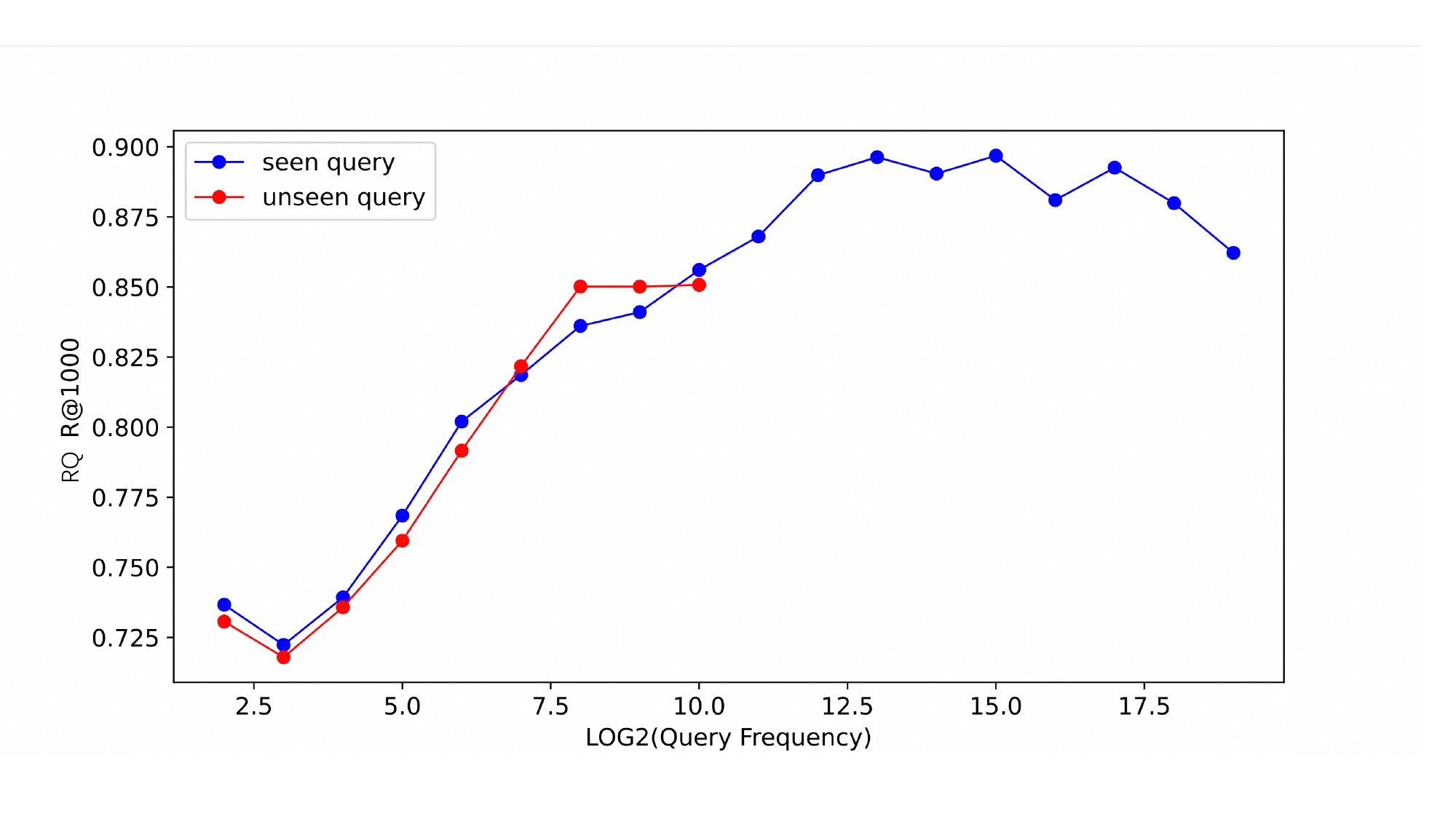}
    \vspace{-20pt}
    \caption{Performance on unseen queries.}
    \label{fig:query_frequency_comparison}
\vspace{-15pt}
\end{wrapfigure}
\textbf{Zero-shot Learning.}
We find that URM effectively adapts to hybrid objectives, such as long-tail item retrieval with a specific query. As illustrated in Table \ref{table:mix_SR_LIR}, combining RQ and LIR objectives not only improves relevance with the query but also increases the proportion of long-tail items.
We also examine URM's performance on seen and unseen queries, noting that unseen queries are often long-tail and have distinct distributions. To enable fair comparison, we plot performance across varying query frequencies. Figure \ref{fig:query_frequency_comparison} illustrates URM’s strong robustness to query input. 

\vspace{-5pt}
\subsection{Online Results}
\vspace{-5pt}
\label{sec:online_results}
\begin{wraptable}{r}{0.32\textwidth}
\vspace{-20pt}
  \ \addtolength{\tabcolsep}{2pt}
\centering
\footnotesize
    \caption{ Online results from April 28, 2025 - May 14, 2025.}
\begin{tabular}{@{}cc@{}}
\toprule
\textbf{Metric} & \textbf{RI}  \\
\midrule
Revenue & +3.01\%  \\
CTR & +0.78\% \\
CVR & +1.24\%  \\
\#Long-tail Items  & +2.23\%  \\ 
\bottomrule
\end{tabular}
    \label{table:online}
\vspace{-5pt}
\end{wraptable}
We evaluate URM on our advertising platform using real traffic. 
A key challenge in retrieval is the discrepancy between offline and online metrics, where offline improvements often don't lead to online gains.
Traditional retrieval models require multiple steps to adjust online retrieval objectives and corresponding results. As a result, optimization for online performance heavily relies on human priors. In contrast, URM can dynamically 
adjust retrieval outputs based on input instruction, which means that we can conduct online A/B tests with various retrieval objectives and user feature combinations to optimize the online performance.
$4$ online metrics are utilized to assess performance: 
advertising revenue (the core metric for our advertising platform), CTR, CVR, and proportion of long-tail items.

In Table \ref{table:online}, we observe a $3.01\%$ growth in advertising revenue. Additionally, the CTR and CVR metric sees a rise of $0.78\%$ and $1.24\%$ respectively, suggesting that URM delivers more precise results to users. Notably, the growth for long-tail items surpasses the overall increase, highlighting that URM is more fair and friendly to long-tail items.

\vspace{-5pt}
\section{Related Work}
\vspace{-5pt}
\paragraph{Multi-Task Learning For Recommendation.}
Multi-task learning methods can be categorized into different parameter sharing methods \cite{CrossSwitch, UberNet, MTAN, ASTMT,Multi_centernet} and optimization strategies \cite{UW, GradNorm, PCGrad}, and have
been widely applied to recommender systems \cite{MMOE, ESMM, PLE}.
However, in practice, it is difficult to ensure that the performance of each task improves after multi-task learning, especially when the number of tasks continues to increase \cite{taskonomy, forkmerge}. 
Thus, it's often necessary to carefully design the model structure based on the data proportions and the task relationships.
This has also led to the isolated states among different recommendation problems, such as multi-scenario modeling \cite{jiang2022adaptivedomainnetworkmultidomain}, multi-objective modeling \cite{ESMM, cite:multi_objective_recommendation}, long-tail recommendation \cite{cite:longtail_recommendation}, etc. 
In addition, search tasks can be viewed as a specialized type of recommendation with explicit query constraints. 
Due to significant distribution discrepancies and limited model capacity, traditional recommendation models cannot simultaneously handle scenarios both including and excluding explicit inputs easily \cite{liu2024unified}.

\vspace{-5pt}
\paragraph{LLMs for Recommendation.}
Inspired by LLMs \cite{cite:GPT,cite:GPT2, GPT3, gpt4}, recent approaches treat recommendation tasks as natural language tasks, generating text results directly through prompting and in-context learning \cite{P5, cite:Uncovering_ChatGPT,IsGPTGood,InstructionFollowing,ilora,TallRec,bao2023bi}. 
However, in real systems, users typically have hundreds or even millions of behaviors, leading to at least tens of thousands of text tokens, which increases inference costs and decreases LLM performance.
Thus, some methods introduce a hierarchical structure that encodes each item's text or image information into item representations and feeds them into LLMs to generate a high-level user representation \cite{HLLM,Jia2024KnowledgeAF,ye2024harnessing,hou2022unisrec,MoRec,TransRec}. 
However, text-based recommendations lack discrimination because many items have similar images, titles, and categories, yet differ in recognition and exposure.
Another approach employs traditional ID embeddings to represent items \cite{CoLLM,E4SRec,LlaRA,Adapt_LLM_Collaborative,cite:IDGenRec}. However, these methods do not fully exploit the advantages of LLMs in recommender systems when dealing with numerous objectives and even unseen objectives, nor do they address the generation problem within a candidate set of tens of millions.

\vspace{-5pt}
\section{Conclusion}
\vspace{-5pt}
In this paper, we answer the questions of why and how LLMs can be used for recommender systems, and introduce Universal Retrieval Model (URM). The value of LLMs in recommender systems lies in their ability to unify various retrieval objectives into a cohesive input-output framework, enabling them to address multiple retrieval objectives simultaneously and prompt tuning the online performance directly. To facilitate the implementation of LLMs in real-world recommender systems, we propose multi-query representation to enable complex user-item relationships modeling, introduce matrix decomposition to improve model learnability, discriminability, and transferability, and incorporate probabilistic sampling to reduce computation costs. Ultimately, URM can adaptively generate a retrieval set from tens of millions of candidates based on any given objective while maintaining latency within tens of milliseconds, achieving significant gains in both offline and online metrics.


\bibliographystyle{plain}
\bibliography{neurips_2025}

\newpage

\appendix

\newpage
\appendix
\onecolumn
\section{Appendix}
\subsection{Limitations}
\label{appendix:limitation}
\textbf{Computational Cost.} In our efforts to deploy LLMs within industrial recommender systems, we've incorporated strategies such as multi-query representation and probabilistic sampling techniques. Despite these advancements, it is important to acknowledge that LLMs inherently increase the costs associated with training and inference compared to traditional recommendation models. Consequently, in our online implementation, we ensure that the daily training data is processed within a 24-hour timeframe by aggregating user actions on an hourly basis and subsequently sampling only 5\% of this aggregated data for supervised fine-tuning (SFT). During the inference phase, we apply a 10-minute window for asynchronous processing. These measures reflect necessary trade-offs between effectiveness and efficiency. While the computational expense is a known limitation of LLMs, our online experiments affirm their substantial value in enhancing recommender systems.

\textbf{Task Versatility.}
We discover that while URM can adapt to certain new objectives, these are often closely linked to the original training objectives. When faced with entirely new objectives, it's challenging for URM to adjust its output based on textual input. 
 The solution to this problem is to increase the amount of training data, enabling URM to learn more mappings from retrieval objectives to results. 
For example, by incorporating search data, we aim to utilize user behavior data associated with diverse queries to enhance URM's ability to generalize to new prompts.

\subsection{Broader Impacts}
\label{appendix:broader_impacts}
URM integrates the vast world knowledge encapsulated in LLMs to offer personalized retrieval to users from a human-centered, interpretable perspective. It is also beneficial in promoting long-tail items, as evidenced by our online experimental results, which demonstrate that our approach is more accommodating to these long-tail items. Consequently, this contributes to more equitable market competition and may help  mitigate monopolistic dominance.


\subsection{Theoretical Guarantee for the Probability Sampling Algorithm}
\label{proof}
A potential assumption of Algorithm \ref{algorithm} is that items close in the $W$ matrix, specifically those that become neighbors, will receive similar scores from URM. A theoretical guarantee is as follows:

\begin{theorem}
    Assuming that the representations of two items $v_1$ and $v_2$ are very similar, i.e., $|| W_{v_1}-W_{v_2}|| \leq \epsilon$. Meanwhile, we apply  bound constraints to the representation $F(u,o)$,
\begin{equation}
     {\overline{F}}(u,o) = \dfrac{{F}(u, o)}{\max(||{F}(u, o)||/B, 1)}
\end{equation}
where  scaling adjustments will only be applied to $F(u, o) $ if its norm exceeds $B$,
then we have  
\begin{equation}
    || \max(W_{v_1}^T \mathbf{\overline{F}}(u,o))-\max(W_{v_2}^T \mathbf{\overline{F}}(u,o))||  \leq \epsilon B 
\end{equation}
\end{theorem}

\begin{proof}
    First, we provide the proof under $M=1$.

\begin{align*}
    & || \max(W_{v_1}^T \mathbf{\overline{F}}(u,o))-\max(W_{v_2}^T \mathbf{\overline{F}}(u,o))|| \\ 
    = &|| W_{v_1}^T {\overline{F}}(u,o)-W_{v_2}^T {\overline{F}}(u,o)|| \\
    \leq 
    &|| W_{v_1}-W_{v_2}|| \cdot || {\overline{F}}(u,o) || \\
    \leq & \epsilon B \\
\end{align*}
This means that as long as the norm of ${\overline{F}}(u,o)$ remain within a constant range, the scores between neighboring items will be similar. 
Then we extend the conclusion to the case when $M>1$. 

Without loss of generality, let $\max(W_{v_1}^T \mathbf{\overline{F}}(u,o))=W_{v_1}^T \overline{F}_i(u,o)$ and $\max(W_{v_2}^T \mathbf{\overline{F}}(u,o))=W_{v_2}^T \overline{F}_j(u,o)$, where $1\leq i,j \leq M$.
By applying bound constraints to the representation, we have 
$||\overline{F}_i(u,o)||\leq B, ||\overline{F}_j(u,o)||\leq B$.

When $i=j$, we have 
\begin{equation}
    || \max(W_{v_1}^T \mathbf{\overline{F}}(u,o))-\max(W_{v_2}^T \mathbf{\overline{F}}(u,o))|| \leq || W_{v_1}-W_{v_2}|| \cdot || {\overline{F}_i}(u,o) || \leq \epsilon B 
\end{equation}

When $i\neq j$, we have
\begin{equation}
    \max(W_{v_1}^T \mathbf{\overline{F}}(u,o))-\max(W_{v_2}^T \mathbf{\overline{F}}(u,o)) = W_{v_1}^T \overline{F}_i(u,o) - W_{v_2}^T \overline{F}_j(u,o)
\end{equation}
By the properties of the max function, we have $W_{v_1}^T \overline{F}_i(u,o) \geq W_{v_1}^T \overline{F}_j(u,o)$ and $W_{v_2}^T \overline{F}_j(u,o) \geq W_{v_2}^T \overline{F}_i(u,o)$, thus we have

\begin{equation}
    W_{v_1}^T \overline{F}_j(u,o) - W_{v_2}^T \overline{F}_j(u,o) \leq  \max(W_{v_1}^T \mathbf{\overline{F}}(u,o))-\max(W_{v_2}^T \mathbf{\overline{F}}(u,o)) \leq W_{v_1}^T \overline{F}_i(u,o) - W_{v_2}^T \overline{F}_i(u,o)
\end{equation}

According to the properties of the norm, we have

\begin{align*}
    &||\max(W_{v_1}^T \mathbf{\overline{F}}(u,o))-\max(W_{v_2}^T \mathbf{\overline{F}}(u,o))|| \\ \leq & \max \{
    || W_{v_1}^T \overline{F}_j(u,o) - W_{v_2}^T \overline{F}_j(u,o) ||,
    || W_{v_1}^T \overline{F}_i(u,o) - W_{v_2}^T \overline{F}_i(u,o) ||\} \\
    \leq & \max \{
    || W_{v_1} - W_{v_2}|| || \overline{F}_j(u,o) ||,
    || W_{v_1} - W_{v_2}|| || \overline{F}_i(u,o) ||
    \} 
     \\
    \leq & \max \{
    \epsilon B,
    \epsilon B
    \} \\
    = &\epsilon B \\
\end{align*}
\end{proof}

Here, $B$ acts as a control hyperparameter for the model’s complexity. A larger $B$ leads to larger model complexity, but also larger inconsistency in $W^T \mathbf{\overline{F}}(u,o)$ between neighbors, thus reducing retrieval precision.   In embedding-based retrieval, user representations are often normalized \cite{EBR_facebook}, which is somewhat similar to the situation where $B = 1$, resulting in very low model expressiveness but very high retrieval precision, suitable for various ANN retrieval methods. In URM, we have $B = 100$, and experiments (Table \ref{table:retrieval}) have shown that retrieval precision still remains high under this setting after multiple samplings.

\newpage
\section{Implementation Details}

\subsection{Dataset Details}
\label{appendix:dataset_details}
Our industrial dataset contains hundreds of millions of samples and more than one billion distinct items in user behavior sequences. 
The candidate set contains tens of millions of items.\footnote{Due to data security concerns, specific numbers cannot be open.}
We use samples from day $1$ to day $T$ for training and the samples of day $T+1$ for testing.
We  construct datasets for typical  objectives, listed as follows:
\begin{itemize}
\setlength{\itemsep}{0pt}
  \item   \textbf{C}lick \textbf{P}rediction \textbf{R}etrieval (\textbf{CPR}):  Positive samples are click behaviors.
  \item    \textbf{R}etrieval for \textbf{S}cene \textbf{A} (\textbf{RSA}), \textbf{S}cene \textbf{B} (\textbf{RSB}), \textbf{S}cene \textbf{C} (\textbf{RSC}): Positive samples are the items exposed to users in $3$ scenarios with distribution shift.
  \item   \textbf{S}erendipity  \textbf{R}etrieval (\textbf{SR}): Retrieve new items for users to mitigate the information cocoon effect. Positive samples are real click behaviors, whose categories have not appeared in user behaviors.
  \item    \textbf{L}ong-term \textbf{R}etrieval (\textbf{LR}): Retrieve items based on long-term behaviors of users. Positive samples are the users' real click behaviors, whose categories only appear in the long-term user behavior sequence.
  \item    \textbf{L}ong-tail \textbf{I}tem \textbf{R}etrieval (\textbf{LIR}): 
    Retrieve long-tail items to mitigate the bias of the recommender system. 
    Positive samples are the items that users have clicked, and the popularity of which is below a certain threshold.
  \item     \textbf{P}urchase \textbf{P}rediction \textbf{R}etrieval (\textbf{PPR}): Positive samples are the items that the users have purchased.
  \item   \textbf{R}etrieval with \textbf{Q}uery (\textbf{RQ}): Given a query, retrieve the most relevant items that the user is most likely to be interested in. Positive samples are the user's click behaviors under a certain query.
\end{itemize}

\subsection{Model Implementation Details}
\label{sec:Implementation}

\textbf{Two-tower Model.} 
Both the user-side features and the item-side features are flattened to $2$-dimensional,  and then fed into $3$-layer MLP networks with layers $[256, 128, 64]$ to obtain the final $64$-dimensional user representations and item representations. The user representations will not be normalized, while the item representations will be normalized.

\textbf{Transformer-based Model.}  After adding the position embedding, the user behavior sequence embedding, task embedding, and query embedding are concatenated and fed into a $5$-layer Transformer network, where each Transformer layer has $4$ heads and the dimension of the FFN layer is $4$ times larger than the input dimension. The embedding at the last position is taken as the user embedding. The item side follows the same as the two-tower model: features are flattened and then processed through a $3$-layer MLP network with dimensions $[256, 128, 64]$ to obtain the final $64$-dimensional embedding, which is then normalized.

\textbf{Attention-DNN.} The user side includes multiple aggregated features and $3$ sequence features of different lengths. The sequence features are used to compute cross-attention with the item features \cite{zhou2018deep}. The item side also includes multiple features. After flattening all features, they are collectively fed into a fully connected network with layers $[256, 128, 64, 1]$ to obtain the final score.
Due to the increase in computation cost for each item, we build an HNSW index \cite{HNSW} on the item representations for faster inference following \cite{zhu2018learning, chen2022approximate}.

\textbf{Attention-DNN + Shared Bottom.} Retain the attention module and all features in Attention-DNN, and use different MLPs for modeling each task separately.

\textbf{Attention-DNN + MMoE.} Retain the attention module and all features in Attention-DNN, and introduce $4$ MLPs to generate the middle representations. The task embedding is fed into the gate network to obtain combination weights, which are then used to adaptively combine the outputs from the MLP.

\textbf{Attention-DNN + PLE.} Retain the attention module and all features in Attention-DNN and adopt $4$ MLPs for shared module and task-specific module separately to generate the middle representations with the dim of $64$. These middle representations of the $4$ shared modules and each task-specific module are then merged using a gate network and fed into each task's $2$-layer MLPs to obtain the final score.

\textbf{URM.} 
The item IDs are first mapped to item embeddings using a distributed hashtable, and then an MLP is used to further map the embeddings into a $4096$-dimensional tensor. 
Simultaneously, the other tokens are also transformed into this dimension by the token embedding table. As depicted in Figure \ref{fig:arch}, once merged with positional embeddings, the input embeddings are fed into the LLM backbone. The hidden state from the final transformer layer will subsequently pass through 2 MLP layers ($U$ and $V$) to yield the final item probability. Between MLP layers, we add RMSNorm \cite{zhang2019rootmeansquarelayer}  to enhance the model expressivenss.


\subsection{Training Details}
\label{sec:Hyperparameters}

For URM, we use AdamW with an initial learning rate of $2e^{-5}$, weight decay of $0.05$, batch size of $8192$, and cosine learning rate decay.
Each batch consists of $10,000$ negative examples, which are sampled from the item candidates according to the $3/4$ power of their occurrence frequency. The sampled training dataset is trained for a single epoch.
To preserve the pre-trained knowledge in LLM, we follow transfer learning \cite{DAN}
and set the learning rate of the pre-trained layers to be $1/10$  of the learning rate of the other layers.
Using Qwen-7B as the LLM backbone on $64$ NVIDIA H20 GPUs, it takes $25$ hours to process $50$ million training examples.
Due to the large volume of data in recommender systems and the high training cost of LLMs, we perform a 5\% sampling of the data when training URM. Meanwhile, we initialize $V_{\text{dis}}$ using the ID Embedding from the two-tower model.

\subsection{Inference Efficiency of URM}
\label{appendix:inference_efficiency}
In the inference stage, our model architecture differs from conventional LLMs in two aspects: (1) item IDs in the input sequence, and (2) probabilistic sampling in the vast item space. 
Specifically, for input item IDs, we utilize distributed hashtable lookup technology to reduce the embedding retrieval latency to within a few milliseconds. For the item sample module, we propose Algorithm \ref{algorithm}, enabling efficient sampling within a ten-million-item database in under 10 milliseconds.
As a result, the negative impact of these two additional modules on performance is relatively minor, resulting in our model being nearly equivalent to an LLM inference with an output token length of $1$. This enables seamless utilization of existing inference frameworks (e.g., vLLM \cite{vllm}) and acceleration techniques (including batching, key-value caching, FlashAttention \cite{dao2022flashattention, dao2023flashattention2}, etc.).

The 99th percentile latency (p99) for various model sizes at different input lengths on NVIDIA H20 GPUs is presented in Table \ref{table:latency_p99}. The results demonstrate that by reducing the model size, efficiency can be further enhanced, which in turn allows for the processing of longer input lengths of tokens. In contrast, generative retrieval methods based on semantic IDs, which require auto-regressively generating multiple tokens, result in online latencies on the order of  $T$.

\begin{table}[htbp]
\addtolength{\tabcolsep}{-3pt}
\centering
\footnotesize
\caption{P99 Latency (ms) for various model sizes at different input lengths.}
\begin{tabular}{@{}lccc@{}}
\toprule
\diagbox{\textbf{Model Size}}{\textbf{Input Length}} & \textbf{256} & \textbf{1024} & \textbf{4096} \\ \midrule
0.5B         & 7.7 & 18   & 65   \\
1.5B         & 18  & 50   & 197  \\
3B           & 34  & 99   & 404  \\
7B           & 63  & 233  & 947  \\ \bottomrule
\end{tabular}
\label{table:latency_p99}
\vspace{-10pt}
\end{table}

\newpage
\section{More Experimental Results}

\subsection{More Experiment Results on Public Dataset}
\label{appendix:more_public_k}
As demonstrated in Table \ref{table:public_k10}, URM exhibits superior performance over the most robust baseline, achieving average relative improvements of $\textbf{49\%}$ in HR@10 and $\textbf{37\%}$ in NDCG@10 on $4$ public datasets, respectively.

\begin{table}[htbp]
\addtolength{\tabcolsep}{-3pt}
\centering
\footnotesize
\caption{Performance on $4$ public datasets($K = 10$).}
\begin{tabular}{@{}l|cc|cc|cc|cc@{}}
\toprule
\multicolumn{1}{c|}{\multirow{2}{*}{\textbf{Methods}}} & \multicolumn{2}{c|}{\textbf{Sports}} & \multicolumn{2}{c|}{\textbf{Beauty}} & \multicolumn{2}{c|}{\textbf{Toys}} & \multicolumn{2}{c}{\textbf{Yelp}} \\ \cmidrule{2-9} 
\multicolumn{1}{c|}{}                                  & HR@10             & NDCG@10          & HR@10             & NDCG@10          & HR@10            & NDCG@10         & HR@10           & NDCG@10         \\ \midrule
HGN                                                    & 0.0313            & 0.0159           & 0.0512            & 0.0266           & 0.0497           & 0.0277          & 0.0326          & 0.0159          \\
GRU4Rec                                                & 0.0204            & 0.0110           & 0.0283            & 0.0137           & 0.0176           & 0.0084          & 0.0263          & 0.0134          \\
Caser                                                  & 0.0194            & 0.0097           & 0.0347            & 0.0176           & 0.0270           & 0.0141          & 0.0253          & 0.0129          \\
BERT4Rec                                               & 0.0191            & 0.0099           & 0.0347            & 0.0170           & 0.0203           & 0.0099          & 0.0090          & 0.0045          \\
FDSA                                                   & 0.0288            & 0.0156           & 0.0407            & 0.0208           & 0.0381           & 0.0189          & 0.0276          & 0.0136          \\
SASRec                                                 & 0.0350            & 0.0192           & 0.0605            & 0.0318           & 0.0698           & 0.0318          & 0.0274          & 0.0136          \\
S3-Rec                                                 & 0.0385            & 0.0204           & 0.0647            & 0.0327           & 0.0700           & 0.0376          & 0.0341          & 0.0168          \\
E4SRec                                                 & 0.0410            & 0.0237           & 0.0758            & 0.0435           & 0.0798           & 0.0479          & 0.0418          & 0.0238          \\
P5                                                     & 0.0460            & 0.0336           & 0.0664            & 0.0429           & 0.0709           & {\ul 0.0587}    & {\ul 0.0703}    & {\ul 0.0445}    \\
TIGER                                                  & 0.0400            & 0.0225           & 0.0648            & 0.0384           & 0.0712           & 0.0432          & -               & -               \\
IDGenRec                                               & {\ul 0.0574}      & {\ul 0.0372}     & {\ul 0.0814}      & {\ul 0.0541}     & {\ul 0.0870}     & 0.0551          & 0.0578          & 0.0404          \\
COBRA                                                  & 0.0434            & 0.0257           & 0.0725            & 0.0456           & 0.0781           & 0.0515          & -               & -               \\
URM                                                    & \textbf{0.1049}   & \textbf{0.0590}  & \textbf{0.1225}   & \textbf{0.0766}  & \textbf{0.1221}  & \textbf{0.0726} & \textbf{0.0866} & \textbf{0.0558} \\ \midrule
\textbf{RI}                                                     & +82.8\%           & +58.6\%          & +50.5\%           & +41.6\%          & +40.3\%          & +23.7\%         & +23.2\%         & +25.4\%         \\ \bottomrule
\end{tabular}
\label{table:public_k10}
\vspace{-10pt}
\end{table}

\subsection{Experiments with More LLM Backbones}
\label{appendix:more_backbone}

\begin{table}[htbp]
\addtolength{\tabcolsep}{-3pt}
\centering
\footnotesize
\caption{Comparison between different LLM backbones on public dataset.}
\begin{tabular}{@{}l|cccc@{}}
\toprule
\multicolumn{1}{c|}{\multirow{2}{*}{\textbf{Methods}}} & \multicolumn{4}{c}{\textbf{Beauty}} \\ \cmidrule{2-5} 
\multicolumn{1}{c|}{}                                  & HR@5             & NDCG@5          & HR@10             & NDCG@10       \\ \midrule
E4SRec                                                 & 0.0525            & 0.0360           & 0.0758            & 0.0435            \\
URM(Qwen-7B)                                                    & 0.0929            & 0.0671           & 0.1225            & 0.0766            \\
URM(LLaMA2-13B)                                                    & 0.0905            & 0.0607           & 0.1339            & 0.0747            \\ \bottomrule
\end{tabular}
\label{table:public-LLM-backbone}
\end{table}

\textbf{Public Dataset.} We conduct evaluations using the LLaMA2-13B model, 
as utilized in E4SRec, on the Amazon Beauty dataset as an example to isolate the effects of the backbone and provide a more aligned comparison. The results are presented in Table \ref{table:public-LLM-backbone}, demonstrating that performance improvements are primarily attributed to our framework rather than deriving from a more advanced backbone.

\textbf{Industrial Dataset.} We conduct experiments on Qwen-1.8B, LLaMA2-13B \cite{Llama2}, and DeepSeek-V2-Lite(16B) \cite{DeepSeekV2-Lite} on the industrial-scale dataset.
As shown in Table \ref{table:taobao-LLM-backbone}, alternative backbone models achieve performance comparable to Qwen-7B. This indicates that the performance improvement is not tied to any specific backbone; rather, it benefits from the versatility and general applicability of our URM framework design.

\begin{table*}[htbp]
\addtolength{\tabcolsep}{5pt}
\centering
\footnotesize
\caption{Comparison between different LLM backbones on industrial dataset (metric: R@1000).}
\resizebox{\textwidth}{!}{
\begin{tabular}{@{}llllllllllll@{}}
\toprule
\textbf{Methods}       & \textbf{CPR} & \textbf{RSA} & \textbf{RSB} & \textbf{RSC} & \textbf{SR} & \textbf{LR} & \textbf{LIR}  & \textbf{PPR} & \textbf{RQ} & \textbf{AVG} \\ \midrule
Qwen-7B  &  \textbf{0.263} &
  \textbf{0.530} &
  \textbf{0.439} &
  0.362 &
  0.093 &
 \textbf{0.285} &
  \textbf{0.240 }&
  0.581 &
  0.835 &
  0.403   \\
Qwen-1.8B & 0.255 & 0.529 & 0.434 & 0.351 & 0.081 & 0.243  & 0.230  & 0.572 & 0.835  & 0.392 \\
LLaMA2-13B & \textbf{0.263} & 0.527 & 0.434 & \textbf{0.37} & \textbf{0.107} & 0.276 & 0.228 & \textbf{0.582} & 0.845 & \textbf{0.404} \\
DeepSeek-V2-Lite(16B) & 0.258 & 0.514 & 0.413 & 0.357 & 0.102 & 0.264 & 0.214 & 0.569 & \textbf{0.851} & 0.394 \\
\bottomrule
\end{tabular}
}
\label{table:taobao-LLM-backbone}
\end{table*}

\subsection{More Experiments on  Multi-Task Learning}
\label{appendix:more_prompt_engineering}

\textbf{STL vs. MTL on URM.} Table \ref{table:taobao_mtl} gives the performance degradation of single-task learning compared to multi-task learning.
Compared to other methods, URM performs better in multi-task settings,
indicating that URM is less prone to seesaw phenomenon. 
Further, Figure \ref{fig:ablation_MTL} compares the performance of single-task and multi-task learning with different amounts of labeled data.
The performance of URM  on the single CPR task can also continuously approach that of multi-task learning as the amount of task-specific data increases. 
However, URM with multi-task learning can converge faster to better performance using much less task-specific data, which has significant value for recommendation scenarios where a large amount of task-specific data is not available.

\begin{table}[htbp]
\addtolength{\tabcolsep}{0.5pt}
\centering
\scriptsize
\caption{Performance on the industrial dataset (metric: R@1000).}
\begin{tabular}{@{}llllllllllll@{}}
\toprule
\textbf{Model} &
  \textbf{Learning Method} &
  \textbf{CPR} &
  \textbf{RSA} &
  \textbf{RSB} &
  \textbf{RSC} &
  \textbf{SR} &
  \textbf{LR} &
  \textbf{LIR} &
  \textbf{PPR} &
  \textbf{RQ} &
  \textbf{AVG} \\ \midrule
\multirow{2}{*}{URM} & STL & 0.137          & 0.226          & 0.219          & 0.089           & 0.042         & 0.089           & 0.107          & 0.427          & 0.668   & 0.223 \\
 & 
  MTL &
  \textbf{0.263} &
  \textbf{0.530} &
  \textbf{0.439} &
  \textbf{0.362} &
  \textbf{0.093} &
  \textbf{0.285} &
  \textbf{0.240} &
  \textbf{0.581} &
  \textbf{0.835} &
  \textbf{0.403} \\ \bottomrule
\end{tabular}
\label{table:taobao_mtl}
\vspace{-5pt}
\end{table}

\begin{figure}[!h]
    \centering
    \vspace{-5pt}
    \includegraphics[width=0.6\linewidth]{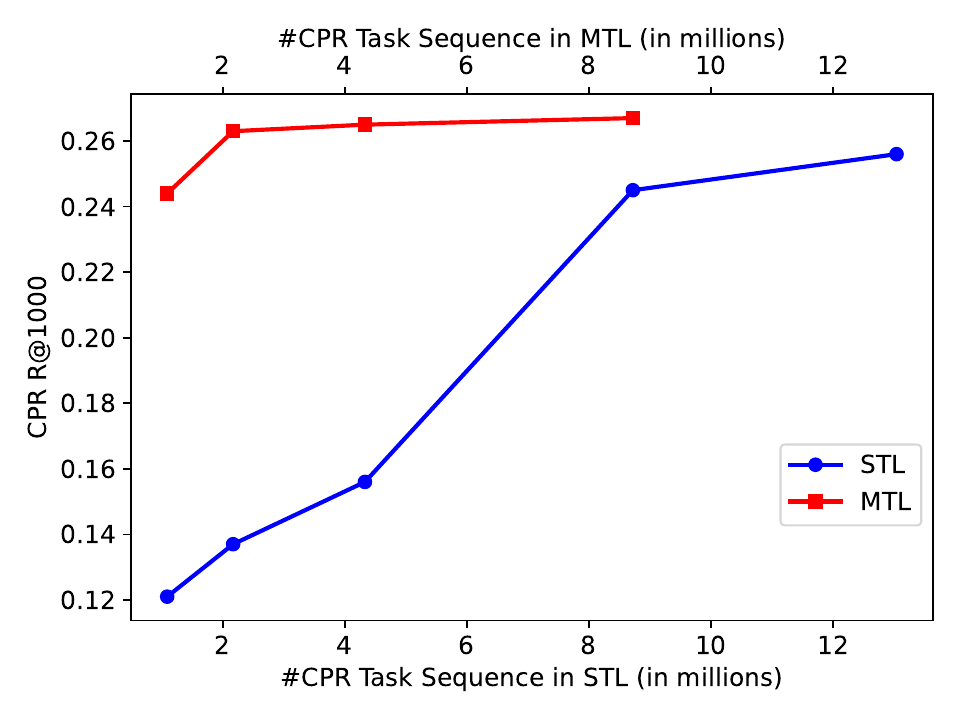}
    \vspace{-15pt}
    \caption{The effectiveness of multi-task learning.}
    \label{fig:ablation_MTL}
\vspace{-5pt}
\end{figure}

\textbf{Long-tail Item Retrieval.} As shown in Table \ref{table:longtail_prompt}, after using the LIR objective, the proportion of long-tail items in the output item set increases by $49.5\%$ relatively.

\begin{table}[!h]
\vspace{-5pt}
\addtolength{\tabcolsep}{5pt}
\centering
\footnotesize
\caption{Performance on Long-tail Item Retrieval.}
\begin{tabular}{@{}llll@{}}
\toprule
\textbf{Objective} &
  \textbf{\begin{tabular}[c]{@{}l@{}}CPR\\ R@1000\end{tabular}} &
  \textbf{\begin{tabular}[c]{@{}l@{}}LIR\\ R@1000\end{tabular}} &
  \textbf{\begin{tabular}[c]{@{}l@{}}Percent of \\ Long-tail Items\end{tabular}} \\ \midrule
CPR                   & \textbf{0.263} & \textbf{0.240}          & 54.6\%            \\
LIR                   & 0.202          & \textbf{0.240}   & \textbf{81.6\%}   \\
\midrule
\textbf{RI} &          -       & +0\% & \textbf{+49.5\%} \\ \bottomrule
\end{tabular}
\label{table:longtail_prompt}
\end{table}

\textbf{Long-term Retrieval.} As shown in Table \ref{table:longterm_prompt}, employing the LR objective results in a relative increase of $96.6\%$ in the proportion of items within the output set that align with the user's long-term interests 
and a relative increase of $36.4\%$ in LR R@1000.

\begin{table}[!h]
\addtolength{\tabcolsep}{5pt}
\centering
\footnotesize
\caption{Performance on Long-term Retrieval.}
\begin{tabular}{@{}llll@{}}
\toprule
\textbf{Objective} &
  \textbf{\begin{tabular}[c]{@{}l@{}}CPR\\ R@1000\end{tabular}} &
  \textbf{\begin{tabular}[c]{@{}l@{}}LR\\ R@1000\end{tabular}} &
  \textbf{\begin{tabular}[c]{@{}l@{}}Percent of Long-\\term Interest\end{tabular}} \\ \midrule
CPR                   & \textbf{0.263} & 0.209          & 32.6\%           \\
LR                   & 0.149          & \textbf{0.285}  & \textbf{64.1\%}  \\
\midrule
\textbf{RI} & -               & \textbf{+36.4\%} & \textbf{+96.6\%} \\ \bottomrule
\end{tabular}
\label{table:longterm_prompt}
\end{table}

\label{appendix:category_diversity}
\textbf{Category Diversity.} In some of the training data, we inject the number of categories in the target set into the objective as shown below. 
 \begin{tcolorbox}[colback=blue!2!white,leftrule=2.5mm,size=title]
\label{obs3}
 Inputs: Please retrieve the top $\{K\}$ categories that users are most likely to be interested in.
\end{tcolorbox}
Then, during the test stage, we observe that by adjusting the size of $K$ in the objective, we could modify the category diversity in the output item set as shown in Table \ref{table:diversity_prompt}.

\begin{table}[!h]
\addtolength{\tabcolsep}{5pt}
\centering
\footnotesize
\caption{Performance on category diversity.}
\begin{tabular}{@{}ccc@{}}
\toprule
\textbf{K} & \textbf{Category Recall@1000} & \textbf{\#Category} \\ \midrule
4          & 0.767               & 74.7                   \\
8          & 0.772               & 80.8                   \\
16         & 0.775               & 94.1                   \\
32         & 0.777               & 104.2                  \\
64         & 0.778               & 113.7                  \\
128        & \textbf{0.780}               & \textbf{117.1}                  \\
\midrule
\textbf{RI}       & \textbf{+1.7\%}      & \textbf{+56.8\%}       \\ \bottomrule
\end{tabular}
\label{table:diversity_prompt}
\end{table}

\subsection{More Experiments on Zero-shot Learning}
\label{appendix:zero_shot_prompt}

\textbf{Hybrid Objectives: Serendipity \& Purchase Prediction Retrieval.}
As shown in Table \ref{table:mix_ER_PP}, when the SR objective and PPR objective are combined, the set produced by the URM aligns better with the purchase objective, thereby improving PPR R@1000 compared to using the SR objective. At the same time, the percentage of new categories in the output set also increases compared to that of the PPR objective. 
This makes the output item set close to the distribution of items that the user has not clicked on before, but is highly likely to purchase next.

\begin{table}[!h]
\addtolength{\tabcolsep}{5pt}
\centering
\footnotesize
\caption{The result of hybridizing the SR and PPR objectives.}
\begin{tabular}{@{}lll@{}}
\toprule
\textbf{Objective} & \textbf{\begin{tabular}[c]{@{}l@{}}PPR\\ R@1000\end{tabular}} & \textbf{\begin{tabular}[c]{@{}l@{}}Percent of \\ New Category\end{tabular}} \\ \midrule
PPR      & \textbf{0.581} & 25.6\%          \\
SR      & 0.411          & \textbf{46.2\%} \\
PPR $\times$ SR & 0.510          & 45.1\%          \\ \bottomrule
\end{tabular}
\label{table:mix_ER_PP}
\end{table}

\textbf{Hybrid Objectives: Long-tail Item Retrieval for Scene A.}
As shown in Table \ref{table:mix_RSA_LIR}, when the RSA objective and LIR objective are combined, the output set is closer to the distribution of scenario A, thereby improving RSA R@1000 compared to using the LIR objective. 
Additionally, there is an increase in the proportion of long-tail items within the output set compared to the RSA objective alone. Consequently, the output item set becomes more representative of the long-tail items that are most likely to occur in scenario A.

\begin{table}[!h]
\addtolength{\tabcolsep}{5pt}
\centering
\footnotesize
\caption{The result of hybridizing the RSA and LIR objectives.}
\begin{tabular}{@{}lll@{}}
\toprule
\textbf{Objective} & \textbf{\begin{tabular}[c]{@{}l@{}}RSA\\ R@1000\end{tabular}} & \textbf{\begin{tabular}[c]{@{}l@{}}Percent of \\ Long-tail Items\end{tabular}} \\ \midrule
RSA       & \textbf{0.530} & 72.6\%          \\
LIR       & 0.410          & \textbf{81.6\%} \\
RSA $\times$ LIR & 0.483          & 81.5\%          \\ \bottomrule
\end{tabular}
\label{table:mix_RSA_LIR}
\end{table}

\textbf{New Objective.}
Further, we validate the zero-shot task transfer performance from the CPR task to the RQ task.
As shown in Figure \ref{fig:compare_rq}, as the number of training samples increases, URM trained only on the CPR task can still learn the mapping between the query and target sets and reach $0.698$ on RQ R@1000! The reason is that our input sequence includes both text and item ID tokens, which allows the URM to align items with the semantic space, even when no query-related tasks exist in the training samples. 
Of course,  the performance of zero-shot task transfer still has a significant gap with supervised training on the corresponding task (MTL vs STL on CPR). 
However, this fully demonstrates that URM inherits the powerful task transfer capabilities of LLM.

\begin{figure}[!h]
\vspace{-5pt}
    \centering
    \includegraphics[width=0.6\linewidth]{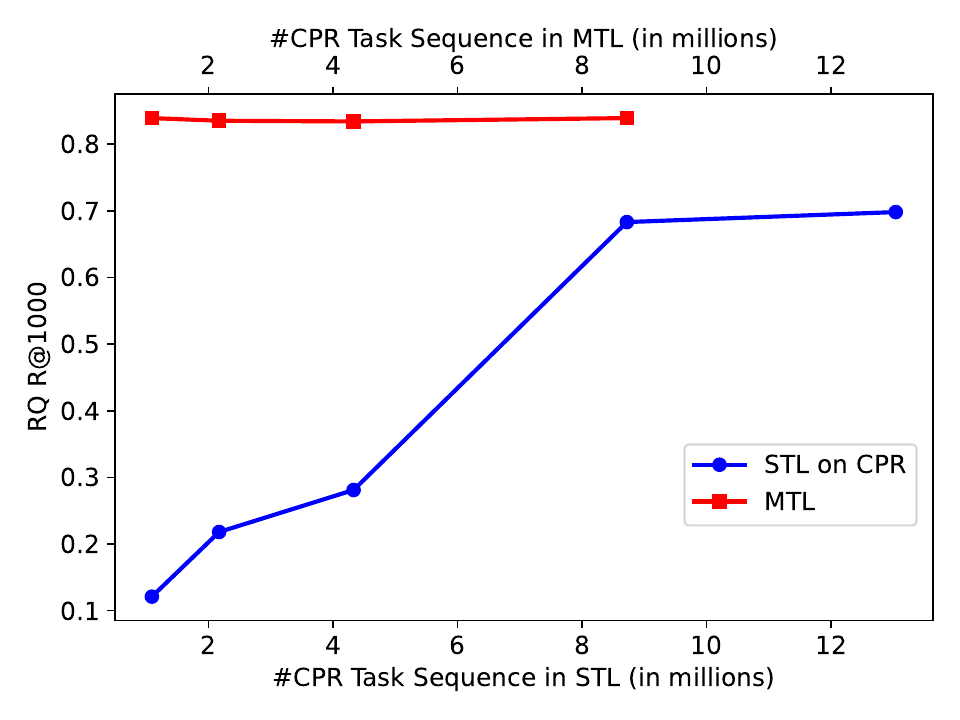}
    \vspace{-15pt}
    \caption{Zero-shot task transfer performance from the CPR task to the RQ task.}
    \label{fig:compare_rq}
\vspace{-15pt}
\end{figure}

\textbf{New Context.}
By modifying the query in the RQ objective to a specific context,  we can model the changes in user behavior over periods, such as during seasonal changes, festival celebrations, or shopping events.
Figure \ref{fig:zero_shot_query} gives two examples.
The fundamental reason why URM can achieve this is that the training data for the RQ (Retrieval with Query) has established a comprehensive mapping between user behaviors, text constraints, and the target item set.
This has two potential advantages. 
\begin{enumerate}
    \item By changing the query to a certain context, we can inject external world knowledge into URM, allowing the results to get ready for potential upcoming events.
    \item This also supports the combination of URM and Chain-of-Thought (CoT) technologies \cite{wei2023chainofthoughtpromptingelicitsreasoning}. Specifically, the LLM generates intermediate results in the form of text through reasoning, and then injects this text as context into URM, thereby producing the final set.
\end{enumerate}

\begin{figure}[htbp]
\vspace{-15pt}
    \centering
    \begin{subfigure}
        \centering
        \includegraphics[width=0.65\linewidth]{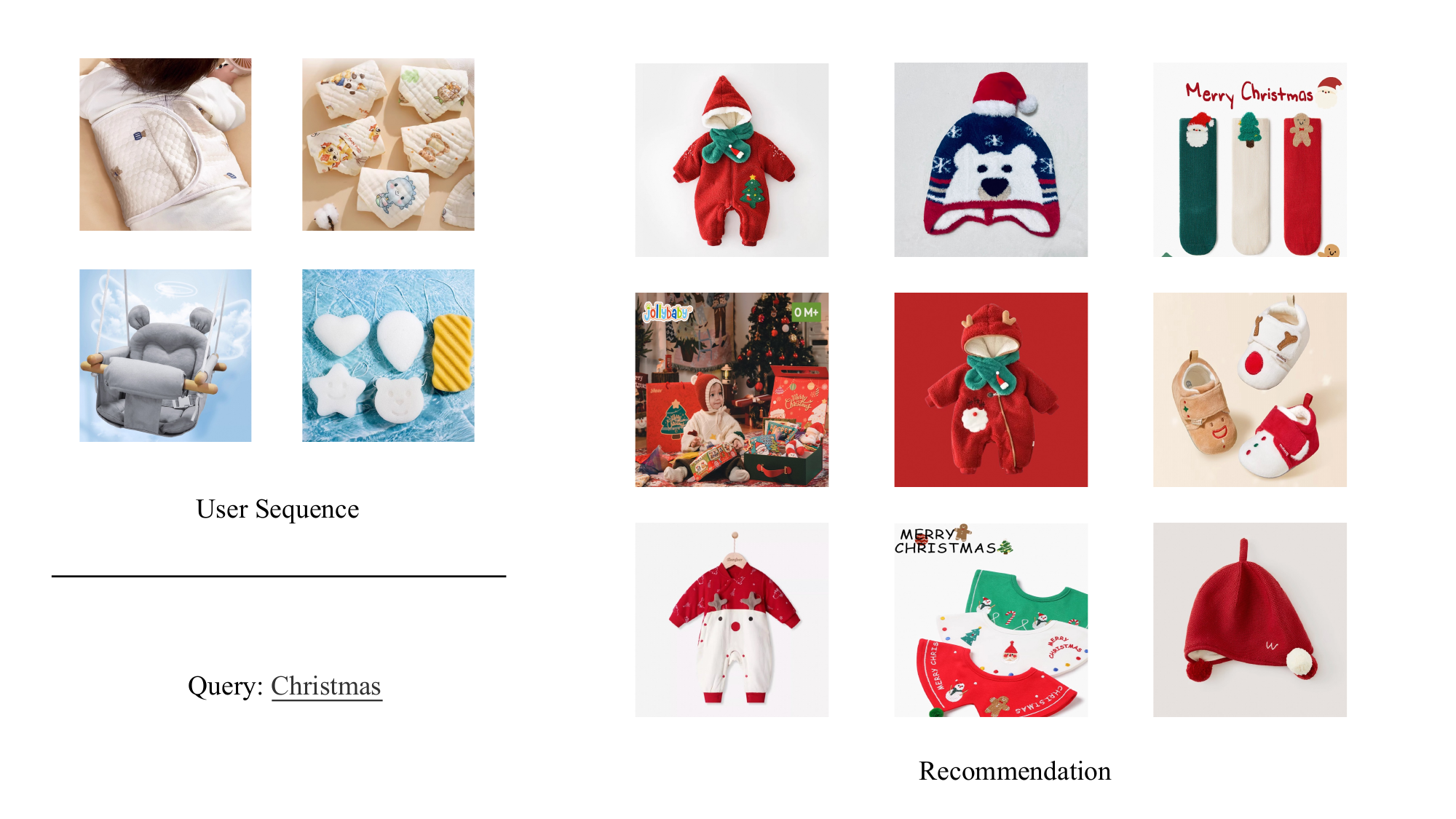}
        \label{fig:query_case1}
    \end{subfigure}
    \vspace{-10pt}
    \begin{subfigure}
        \centering
        \includegraphics[width=0.65\linewidth]{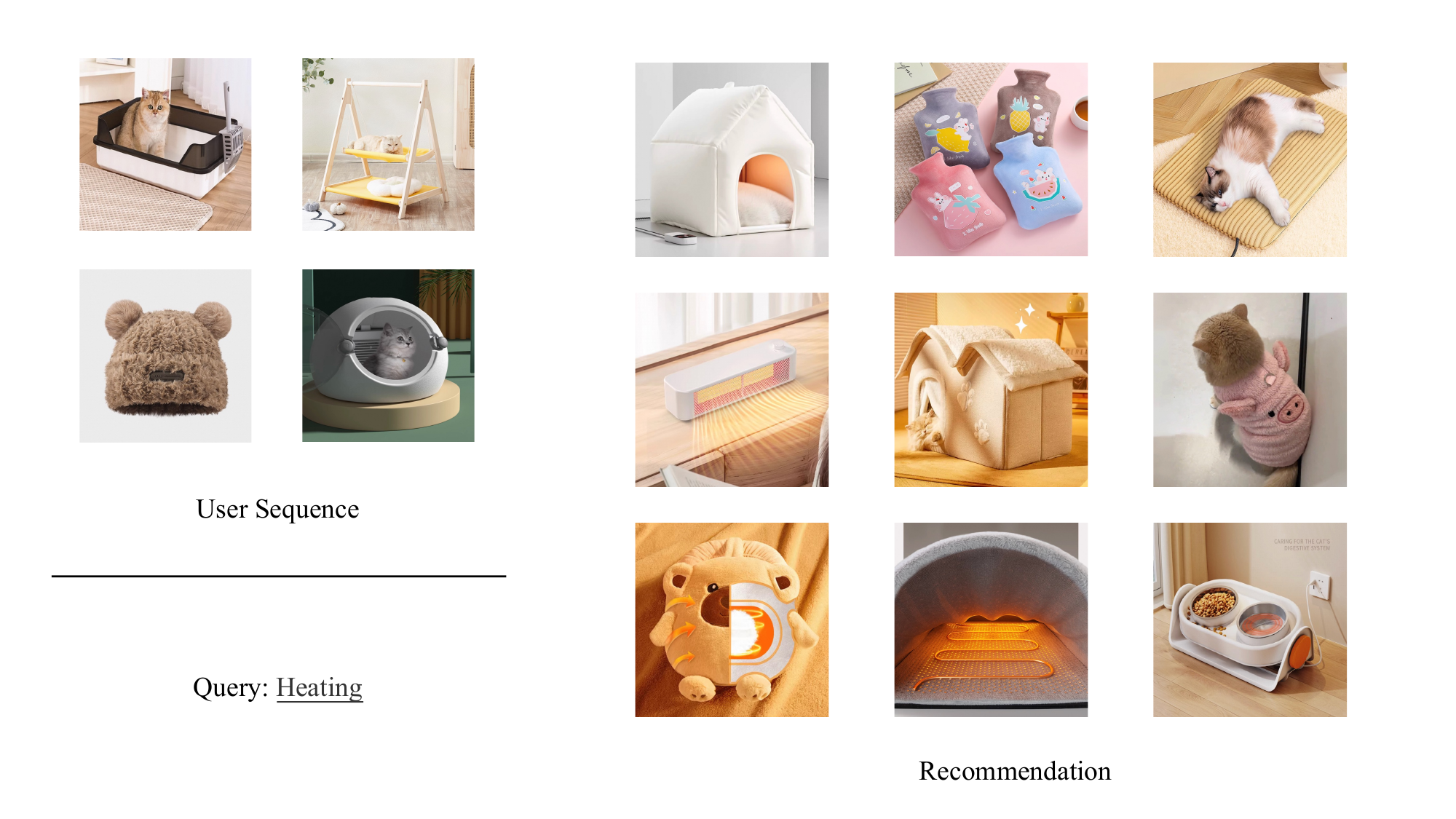}
        \label{fig:query_case2}
    \end{subfigure}
    \caption{Zero-shot task transfer to a new context. Christmas is a festive celebration, and heating is a common need in winter. When these two contexts are injected as queries into the objective, the model can retrieve items based on the user's historical behavior and the given context simultaneously.}
    \label{fig:zero_shot_query}
\end{figure}

\subsection{More Experiments on Multi-Query Representation}
\label{appendix:multi_query_representation}
\textbf{Visualization of User Representations from Different Query Tokens.} To further validate the role of query tokens, we categorize the output $1000$ items from URM into the user representations that have the largest inner product with it. The items under $3$ user representations are shown in Figure \ref{fig:head_case}. It can be observed that different low-dimensional user representations have captured different user interests, which enhances the representation ability of URM for the target item set.

\begin{figure}[!h]
    \centering
    \includegraphics[width=0.7\linewidth]{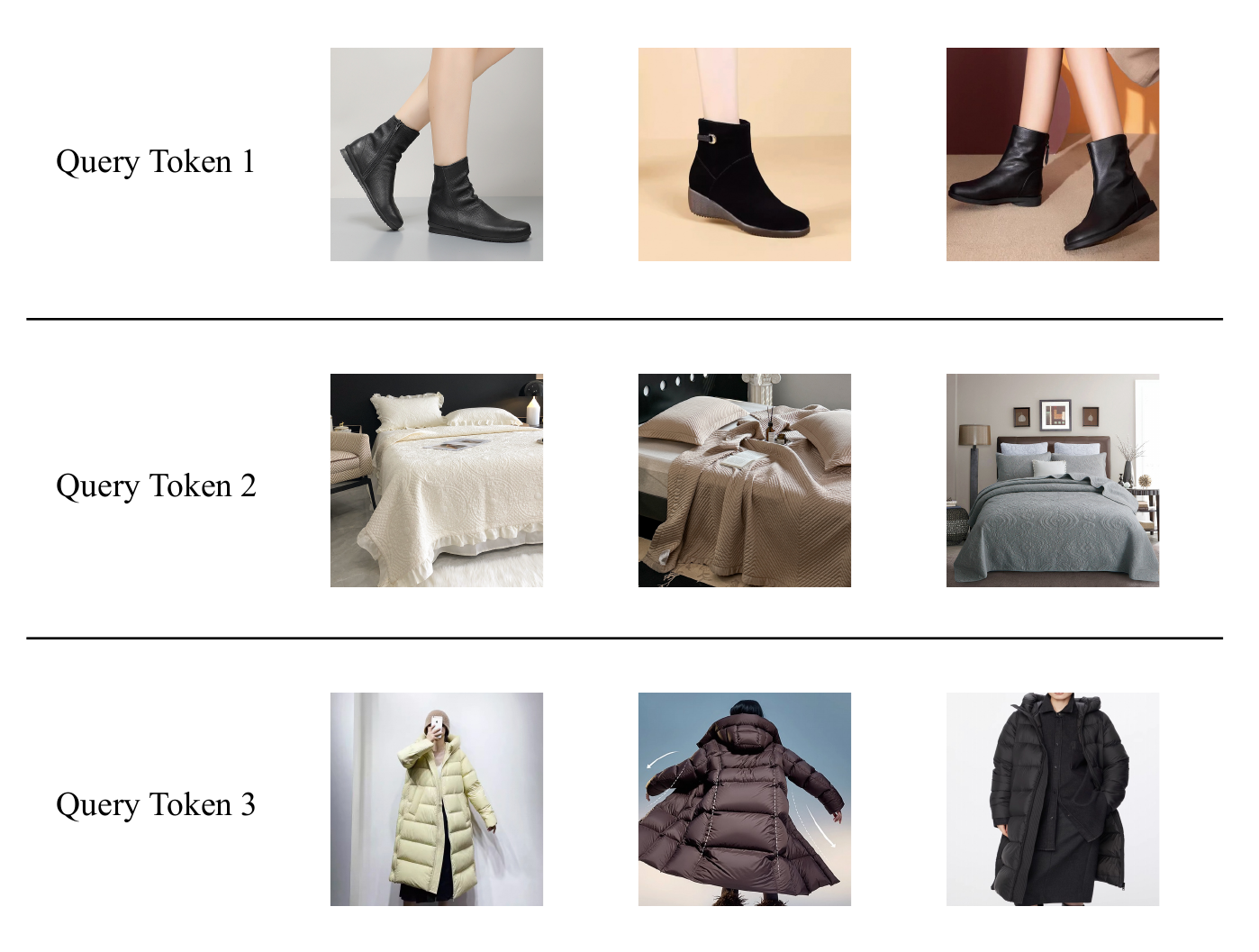}
    \caption{Visualization of user representations from different query tokens. Each token captures distinct facets of user interest, enabling LLM to jointly express different aspects of the target item set.}
    \label{fig:head_case}
\end{figure}

\textbf{Activation of Query Tokens on Different Tasks.} 
Figure \ref{fig:head_proportion} shows the activation proportion for $8$ different query tokens in the CPR and RQ tasks. A token is activated when the dot product of the user representation corresponding to that query token and the target item representation is used for the final output.
We find that different query tokens exhibit variation across different tasks. 

\begin{figure}[!h]
    \centering
    \includegraphics[width=0.6\linewidth]{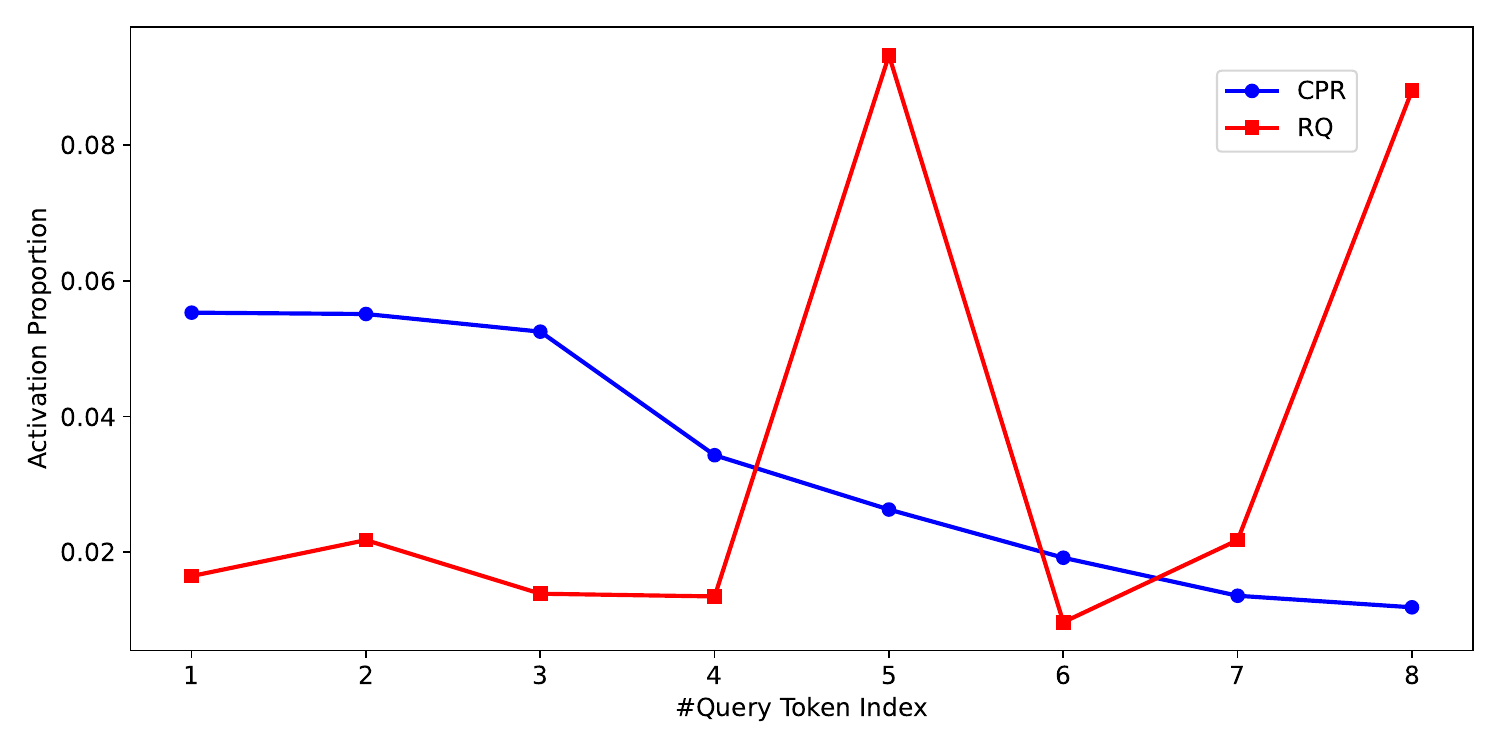}
    \caption{Activation proportion of different query tokens in CPR and RQ tasks.}
    \label{fig:head_proportion}
\end{figure}

\textbf{The Necessity of Multi-Query Representation Format.}
URM employs multiple query tokens to generate several 128-dimensional output user representations, which are subsequently computed alongside the 128-dimensional item representation. To explore the effect of the multi-query representation format, we conduct experiments using higher embedding dimensions and a single token, with the results presented in Table \ref{table:multi_token}. On one hand, when a single token is used to derive a 4096-dimensional user representation, which is then split into 32 separate 128-dimensional representations, the CPR recall drops significantly from 0.248 to 0.163. This highlights the advantage of generating multiple user representations with different query tokens, even when the formal dimension remains constant. On the other hand, increasing the item representation dimension from 128 to 4096 (resulting in logits computed in a higher-dimensional space) does not enhance performance compared to the logits derived from multiple tokens, yielding a mere CPR recall of 0.203. These findings confirm that the multi-query representation design doesn't solely benefit from its higher dimensionality; rather, it demonstrates superior effectiveness in capturing the diverse interests of users within this specific format.

\begin{table}[!h]
\addtolength{\tabcolsep}{10pt}
\centering
\footnotesize
\caption{The effect of multiple query tokens (metric: R@1000).}
\begin{tabular}{@{}lll@{}}
\toprule
\textbf{Methods} & \begin{tabular}[c]{@{}l@{}}\textbf{CPR}\end{tabular} & \begin{tabular}[c]{@{}l@{}}\textbf{RQ} \end{tabular} \\ \midrule
32 Tokens with 128-d Embeddings   & \textbf{0.248}                                                    & \textbf{0.835}                                                     \\
1 Token with 4096-d User Representation               & 0.163                                                    & 0.774       \\
1 Token with 4096-d Item Representation & 0.203 & 0.781
\\ \bottomrule
\end{tabular}
\label{table:multi_token}
\end{table}

\subsection{Visualization Comparison of Different Item Representation}
\label{appedix:visualize_item}
To visualize different item representations, we use the nearest neighbor retrieval approach to find similar items to a given item under a specific representation.
As shown in Figure \ref{fig:emb_case1},\ref{fig:emb_case2},\ref{fig:emb_case3}, 
$V_{\text{trans}}$ mainly captures similarity relationships, whereas $V_{\text{dis}}$ are more centered on co-occurrence relationships. The fusion representations fall somewhere in between these two.

    
    
    

\begin{figure}[!h]
    \centering
    \includegraphics[width=0.8\linewidth]{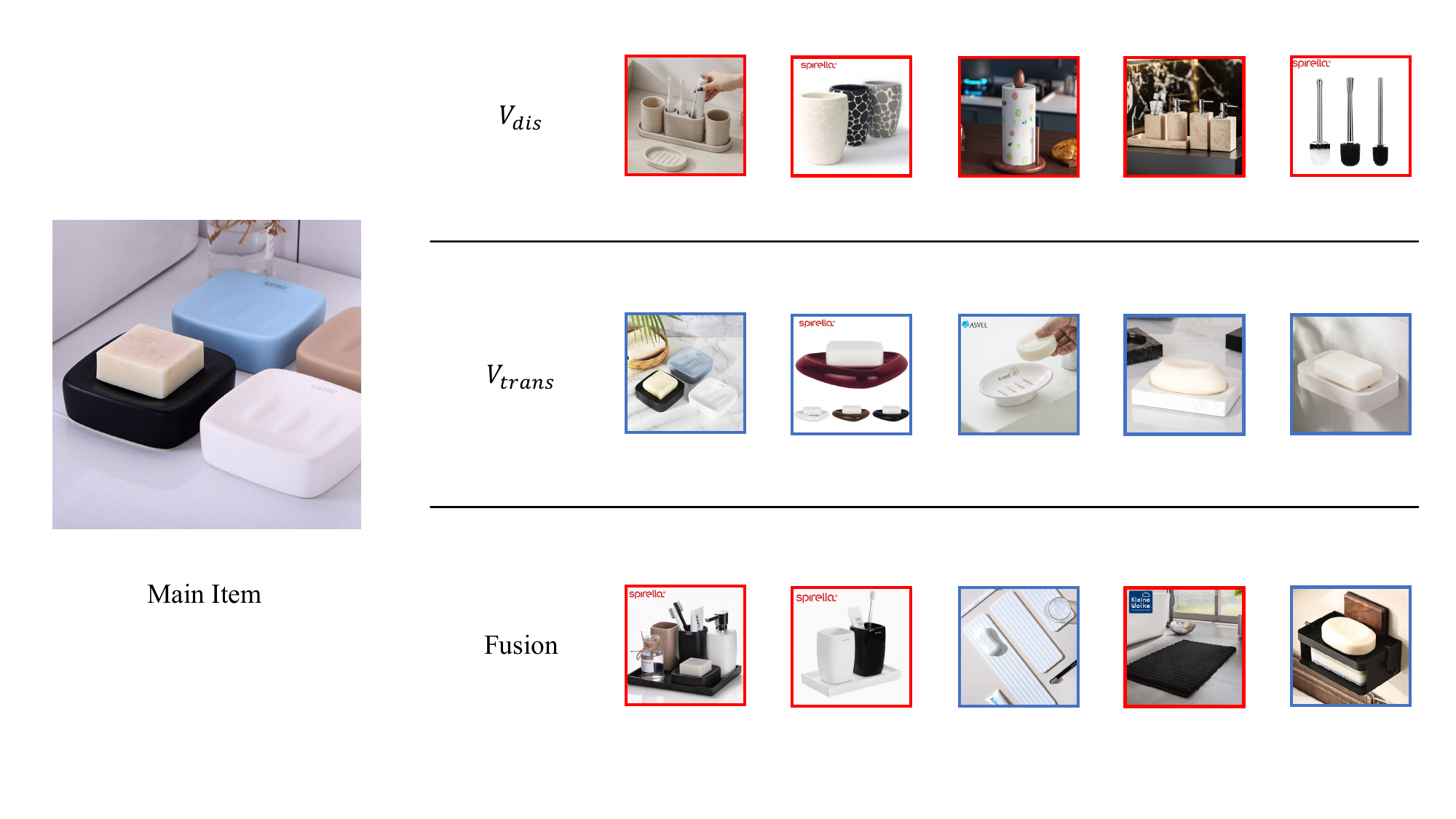}
    \caption{Visualization of different item representations (case 1).}
    \label{fig:emb_case1}
\end{figure}

\begin{figure}[!h]
    \centering
    \includegraphics[width=0.8\linewidth]{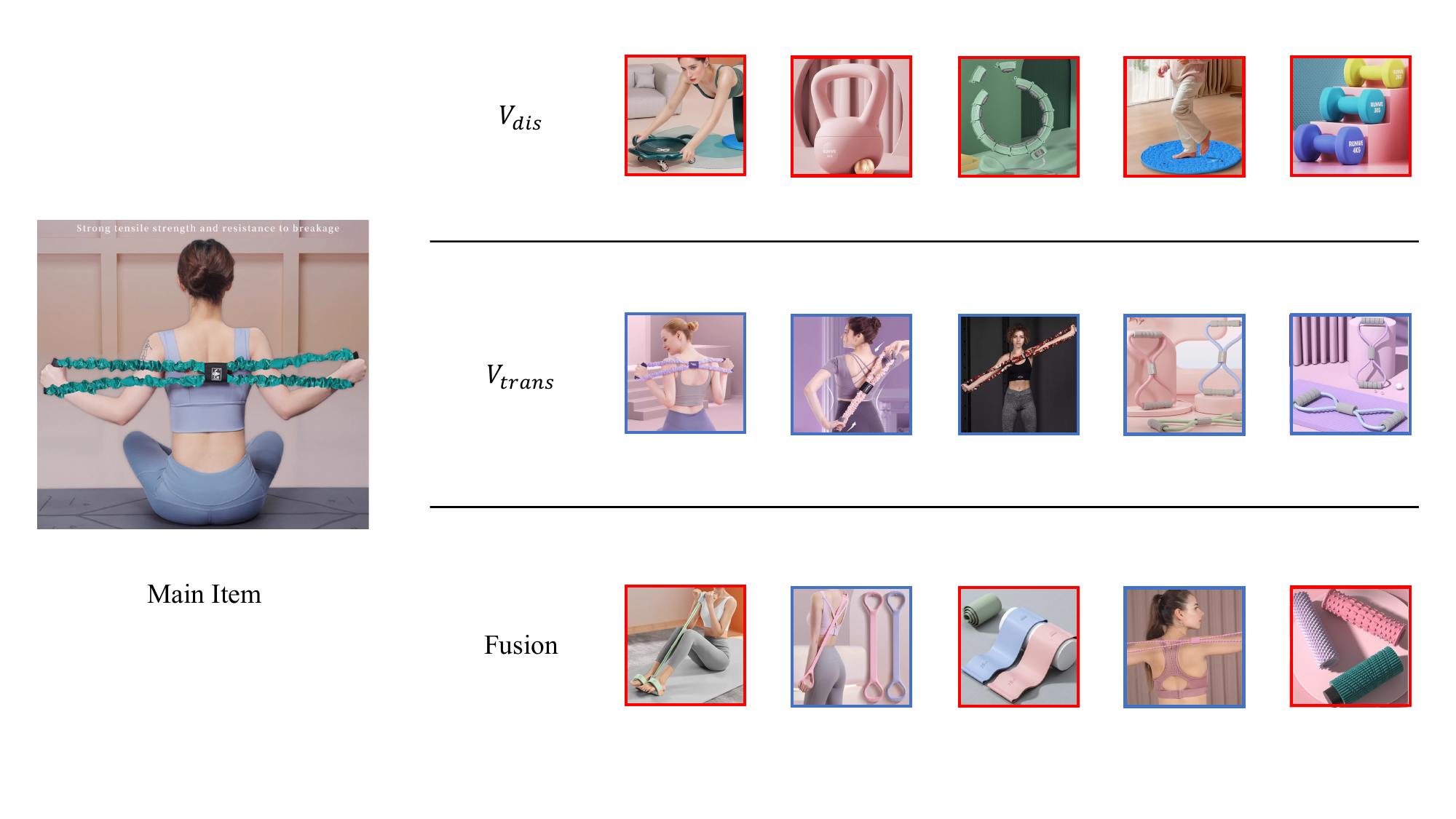}
    \caption{Visualization of different item representations (case 2).}
    \label{fig:emb_case2}
\end{figure}

\newpage

\begin{figure}[!h]
    \centering
    \includegraphics[width=0.8\linewidth]{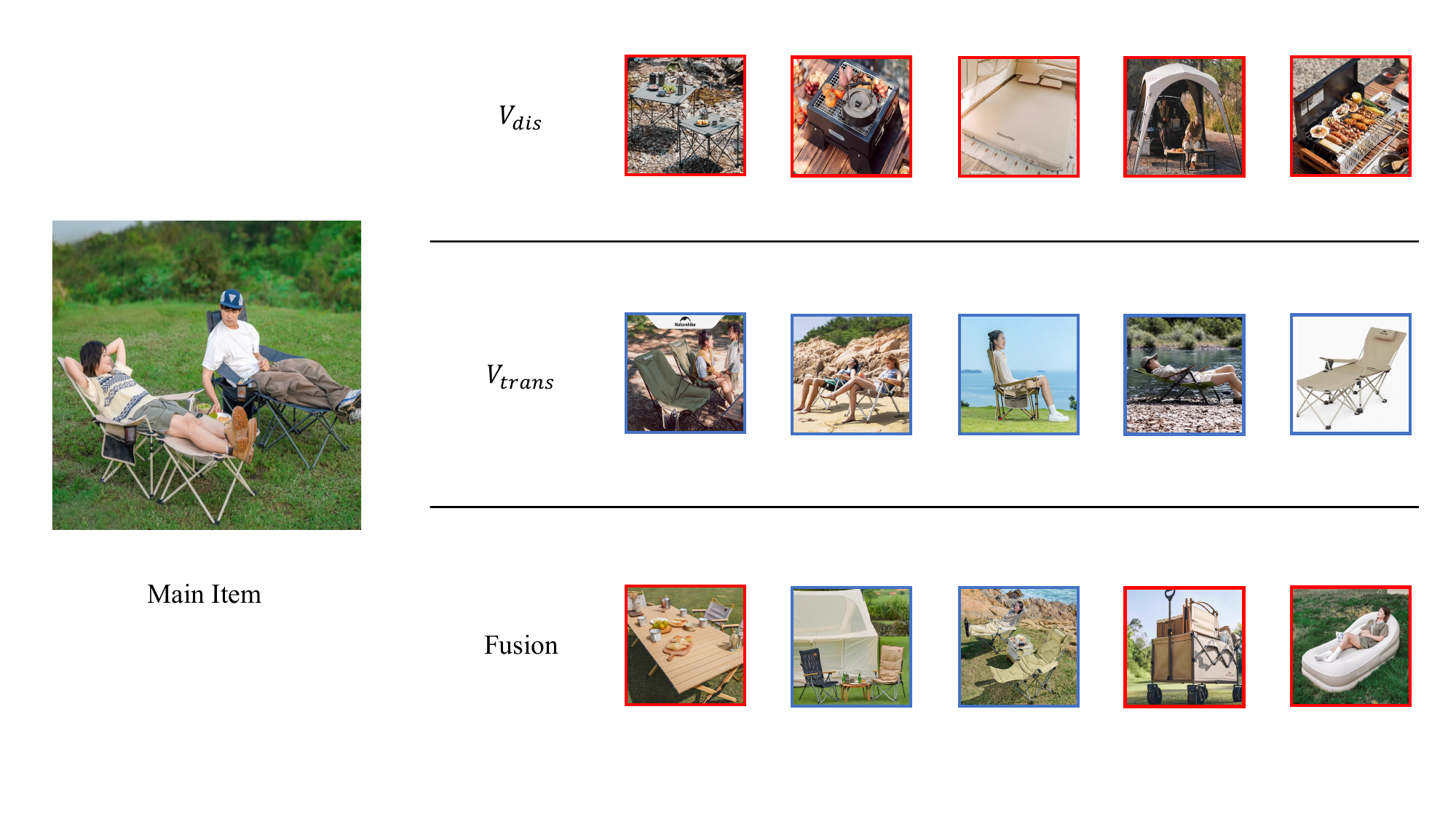}
    \caption{Visualization of different item representations (case 3).}
    \label{fig:emb_case3}
\end{figure}

\newpage
\newpage
\section*{NeurIPS Paper Checklist}
\begin{enumerate}

\item {\bf Claims}
    \item[] Question: Do the main claims made in the abstract and introduction accurately reflect the paper's contributions and scope?
    \item[] Answer: \answerYes{} 
    \item[] Justification: The abstract and introduction clearly state our
contributions -- introducing the Universal Retrieval Model (URM), enhancing model expressiveness, learnability, and efficiency of generative retrieval, and achieving improvements in both offline and online metrics.


\item {\bf Limitations}
    \item[] Question: Does the paper discuss the limitations of the work performed by the authors?
    \item[] Answer:  \answerYes{} 
    \item[] Justification: Due to space constraints, we have placed the limitations section in the Appendix \ref{appendix:limitation}. 

\item {\bf Theory assumptions and proofs}
    \item[] Question: For each theoretical result, does the paper provide the full set of assumptions and a complete (and correct) proof?
    \item[] Answer: \answerYes{} 
    \item[] Justification: We provide detailed proof in Appendix \ref{proof} for the assumption in Algorithm \ref{algorithm}.

\item {\bf Experimental result reproducibility}
    \item[] Question: Does the paper fully disclose all the information needed to reproduce the main experimental results of the paper to the extent that it affects the main claims and/or conclusions of the paper (regardless of whether the code and data are provided or not)?
    \item[] Answer: \answerYes{} 
    \item[] Justification: In Algorithm \ref{algorithm}, we provide a pseudo-code for our probabilistic sampling methods.
    In Section \ref{sec:open_dataset_experiment},\ref{sec:industrial_experiment}, Appendix \ref{appendix:dataset_details},\ref{sec:Implementation},\ref{sec:Hyperparameters}, we provide all the details to reproduce our experiment.

\item {\bf Open access to data and code}
    \item[] Question: Does the paper provide open access to the data and code, with sufficient instructions to faithfully reproduce the main experimental results, as described in supplemental material?
    \item[] Answer: \answerNo{} 
    \item[] Justification: Upon acceptance, we will release the code related to the public datasets.
    Due to data security concerns, our code for the industrial online and offline experiments cannot be made public.

\item {\bf Experimental setting/details}
    \item[] Question: Does the paper specify all the training and test details (e.g., data splits, hyperparameters, how they were chosen, type of optimizer, etc.) necessary to understand the results?
    \item[] Answer: \answerYes{} 
    \item[] Justification: The experimental setting is presented Section \ref{sec:open_dataset_experiment} and \ref{sec:industrial_experiment}.
       The full details can be found in Appendix \ref{sec:Implementation} and \ref{sec:Hyperparameters}.

\item {\bf Experiment statistical significance}
    \item[] Question: Does the paper report error bars suitably and correctly defined or other appropriate information about the statistical significance of the experiments?
    \item[] Answer: \answerNo{} 
    \item[] Justification: 
    Given the high computational cost, we did not include detailed statistical significance measures for our experiments. Nonetheless, our method consistently outperforms baseline methods across multiple tasks using both public and industrial datasets, showing marked improvements in online performance. This serves as compelling evidence of our approach's effectiveness.

\item {\bf Experiments compute resources}
    \item[] Question: For each experiment, does the paper provide sufficient information on the computer resources (type of compute workers, memory, time of execution) needed to reproduce the experiments?
    \item[] Answer: \answerYes{} 
    \item[] Justification: The computer resources for our experiment are given in Appendix \ref{sec:Hyperparameters} and \ref{appendix:inference_efficiency}.
    
\item {\bf Code of ethics}
    \item[] Question: Does the research conducted in the paper conform, in every respect, with the NeurIPS Code of Ethics \url{https://neurips.cc/public/EthicsGuidelines}?
    \item[] Answer: \answerYes{} 
    \item[] Justification: Our manuscript adheres to the NeurIPS Code of Ethics by ensuring equitable treatment, safeguarding data privacy, and so on.

\item {\bf Broader impacts}
    \item[] Question: Does the paper discuss both potential positive societal impacts and negative societal impacts of the work performed?
    \item[] Answer: \answerYes{} 
    \item[] Justification: We discuss potential positive societal impacts in Appendix \ref{appendix:broader_impacts}, and we foresee no negative societal impacts resulting from our research. There may be concerns regarding the use of sensitive user or item information; however, we address this by omitting such details and utilizing non-sensitive attributes in practice.
    
\item {\bf Safeguards}
    \item[] Question: Does the paper describe safeguards that have been put in place for responsible release of data or models that have a high risk for misuse (e.g., pretrained language models, image generators, or scraped datasets)?
    \item[] Answer: \answerNA{} 
    \item[] Justification: Our paper poses no such risks.

\item {\bf Licenses for existing assets}
    \item[] Question: Are the creators or original owners of assets (e.g., code, data, models), used in the paper, properly credited and are the license and terms of use explicitly mentioned and properly respected?
    \item[] Answer: \answerYes{} 
    \item[] Justification: We have cited all the original papers that produced the code package or dataset.

\item {\bf New assets}
    \item[] Question: Are new assets introduced in the paper well documented and is the documentation provided alongside the assets?
    \item[] Answer: \answerNA{} 
    \item[] Justification: Our paper does not release new assets.

\item {\bf Crowdsourcing and research with human subjects}
    \item[] Question: For crowdsourcing experiments and research with human subjects, does the paper include the full text of instructions given to participants and screenshots, if applicable, as well as details about compensation (if any)? 
    \item[] Answer: \answerNA{}{} 
    \item[] Justification: Our paper does not involve crowdsourcing nor research with human subjects.

\item {\bf Institutional review board (IRB) approvals or equivalent for research with human subjects}
    \item[] Question: Does the paper describe potential risks incurred by study participants, whether such risks were disclosed to the subjects, and whether Institutional Review Board (IRB) approvals (or an equivalent approval/review based on the requirements of your country or institution) were obtained?
    \item[] Answer: \answerNA{} 
    \item[] Justification: Our paper does not involve crowdsourcing nor research with human subjects.

\item {\bf Declaration of LLM usage}
    \item[] Question: Does the paper describe the usage of LLMs if it is an important, original, or non-standard component of the core methods in this research? Note that if the LLM is used only for writing, editing, or formatting purposes and does not impact the core methodology, scientific rigorousness, or originality of the research, declaration is not required.
    \item[] Answer: \answerYes{} 
    \item[] Justification: Section \ref{sec: representations}
 provides the inputs and outputs of the LLMs and a detailed explanation of the modifications we made to the existing LLMs.

\end{enumerate}

\end{document}